\documentclass[twocolumn, 10pt]{article}

\input{arxiv_load_packages}
\newcommand{\E}{\mathbb{E}}
\newcommand{\V}{\mathbb{V}}
\newcommand{\R}{\mathbb{R}}
\newcommand{\C}{\mathbb{C}}
\newcommand{\Z}{\mathbb{Z}}

\DeclareMathOperator*{\argmax}{arg\,max} 
\DeclareMathOperator*{\argmin}{arg\,min} 

\def\independenT#1#2{\mathrel{\setbox0\hbox{$#1#2$}%
\copy0\kern-\wd0\mkern4mu\box0}}

\newcommand{\vnorm}[1]{\lVert #1 \rVert}

\newcommand{\KL}[2]{\mathcal{D}_{KL}\left(#1 \mid \mid #2 \right)} 
\newcommand{\SymKLexp}[2]{\frac{\KL{p}{q} + \KL{q}{p}}{2}}

\newfloat{myalgo}{tbhp}{mya}

{\begin{myalgo}[#1]
\centering
\begin{minipage}{0.9\textwidth}
\begin{algorithm}[H]}%
{\end{algorithm}
\end{minipage}

\end{myalgo}}

\newcommand{\transpose}[1]{#1^\top}
\newcommand{\interest}[1]{\iota \left(  #1 \right)}

\newcounter{bar}
\newcommand{\foo}{%
\stepcounter{bar}%
\thebar}

\newtheorem{theorem}{Theorem}[section]

\newtheorem{corollary}[theorem]{Corollary}

\newtheorem{definition}[theorem]{Definition}

\newtheorem{properties}[theorem]{Properties}
\newtheorem{proposition}[theorem]{Proposition}

\newtheorem{remark}[theorem]{Remark}

\usepackage[authoryear,sort&compress,longnamesfirst]{natbib}
\graphicspath{{images/}}
\DeclareGraphicsExtensions{.pdf}

\newcommand{\Title}{Forecastable Component Analysis (ForeCA)\footnote{Appears in Proceedings of ICML 2013.}}
\newcommand{\Author}{Georg M.\ Goerg \\ \url{gmg@stat.cmu.edu}, \url{www.stat.cmu.edu/\textasciitilde gmg}}
\newcommand{\Affil}{Carnegie Mellon University, Department of Statistics}
\newcommand{\City}{Pittsburgh, PA 15213, USA}

\title{\Title}
\author{\Author \\ \Affil \\ \City}
\date{\today}

\usepackage[margin = 2.25cm, top = 2cm]{geometry}

\begin{document}

\pagenumbering{arabic}

\maketitle

\begin{abstract}
I introduce \textbf{Fore}castable \textbf{C}omponent \textbf{A}nalysis (ForeCA), a novel dimension reduction technique for temporally dependent signals.  Based on a new \emph{forecastability} measure, ForeCA finds an optimal transformation to separate a multivariate time series into a forecastable and an orthogonal white noise space.  I present a converging algorithm with a fast eigenvector solution. Applications to financial and macro-economic time series show that ForeCA can successfully discover informative structure, which can be used for forecasting as well as classification.

The R package \href{cran.r-project.org/web/packages/ForeCA}{ForeCA} accompanies this work and is publicly available on \href{cran.r-project.org}{CRAN}.
\end{abstract}

\section{Introduction}
With the rise of high-dimensional datasets it has become important to perform dimension reduction (DR) to a lower dimensional representation of the data.  For simplicity we consider linear transformations $\mathbf{W} \in \R^{k \times n}$, which map an $n$-dimensional $\mathbf{X}$ to a $k \leq n$ dimensional $\mathbf{S} = \mathbf{W} \mathbf{X}$.   Typically, the transformed data should be somewhat ``interesting''; there is no point in transforming $\mathbf{X}$ to an arbitrary $\mathbf{S}$ that is less useful, meaningful, etc. Let $\interest{\mathbf{S}}$ measure ``interestingness'' of $\mathbf{S}$. DR can then be set up as an optimization problem
\begin{align}
\label{eq:max_problem}
\widehat{\mathbf{w}}_j &= \argmax_{\mathbf{w} \in \R^{n \times 1}} \interest{ \transpose{\mathbf{w}} \mathbf{X}}, \quad j = 1, \ldots, k, \\
\label{eq:orthogonal_signal_j}
\text {subject to }& \transpose{\mathbf{w}_j} \mathbf{X} \, \bot \,  \lbrace \transpose{\mathbf{w}_1} \mathbf{X}, \ldots, \transpose{\mathbf{w}_{j-1}} \mathbf{X} \rbrace,
\end{align}
where \eqref{eq:orthogonal_signal_j} is a common DR constraint, which makes $\mathbf{S}_j = \transpose{\mathbf{w}}_j \mathbf{X}$ orthogonal (uncorrelated) to previously obtained signals.

For example, principal component analysis (PCA) keeps large variance signals \citep{Jolliffe02_PCAbook} -- $\interest{X} = \E (X - \E X)^2$ in  \eqref{eq:max_problem}; independent component analysis (ICA) recovers statistically independent signals \citep{Hyvarinen00_ICAtutorial}; slow feature analysis (SFA) \citep{Wiskott02_SFA} finds ``slow'' signals and is equivalent to maximizing the lag $1$ autocorrelation coefficient.

DR techniques are often applied to multivariate time series $\mathbf{X}_t$, hoping that forecasting on the lower-dimensional space $\mathbf{S}_t$ is more accurate, simpler, more efficient, etc.  Standard DR techniques such as PCA or ICA, however, do not explicitly address forecastability of the sources. For example, just because a signal has high variance does not mean it is easy to forecast. 

\begin{figure*}[!t]
  \begin{center}
\makebox{\includegraphics[width=0.85\textwidth, trim = 0cm 0cm 0cm 0cm, clip = true]{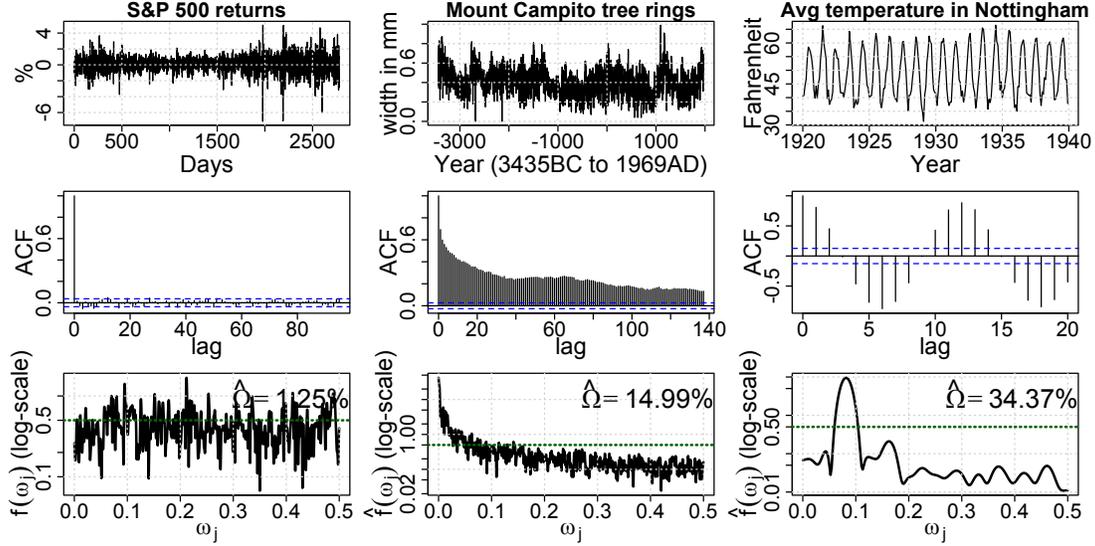}}
  \end{center}
\caption{\label{fig:real_world_examples} Observations (top); sample ACF $\widehat{\rho}(k)$ (middle); smoothed WOSA spectral density estimate (bottom). From left to right: i) S\&P 500 daily returns; ii) Mount Campito tree ring series; iii) monthly mean temperatures in Nottingham. Data publicly available in R packages: \texttt{SP500} in \texttt{MASS}; \texttt{camp} in \texttt{tseries}; \texttt{nottem} in \texttt{datasets}.}
\end{figure*}

Thus let's define interesting as being predictable. Forecasting is not only good for its own sake (finance, economics), but even when future values are not immediately interesting, signals that do have predictive power exhibit non-trivial structure by definition -- and are thus easier to interpret. 
For example, the time series in Fig.\ \ref{fig:real_world_examples} are ordered from least (S\&P500 daily returns) to most forecastable (monthly temperature in Nottingham) according to the ForeCA forecastability measure $\Omega(x_t)$ I propose in Definition \ref{def:Omega} below. And indeed moving from left to right they exhibit more structure.

The main contributions of this work are
\begin{inparaenum}[i)]
\item a model-free, comparable measure of forecastability for (stationary) time series (Section \ref{sec:measuring_forecastability}),
\item a novel data-driven DR technique, ForeCA, that finds forecastable signals,
\item an iterative algorithm that provably converges to (local) optima using fast eigenvector solutions (Section \ref{sec:maximize_Omega}), and 
\item applications showing that ForeCA outperforms traditional DR techniques in finding low-dimensional, forecastable subspaces, and that it can also be used for time series classification (Section \ref{sec:applications}).
\end{inparaenum}
Related work will be reviewed in Section \ref{sec:related_work}.
 
All computations and simulations were done in R \citep{R10}.

\section{Time Series Preliminaries}
\label{sec:TSA}
Let $y_t$ be a univariate, second-order stationary time series with mean $\E y_t = \mu_y < \infty$, variance $\V y_t = \sigma_y^2$, and autocovariance function (ACVF)
\begin{equation}
\label{eq:ACVF_univariate}
\gamma_y(k) = \E (y_t - \mu_y) \left( y_{t-k} - \mu_y \right), \quad k \in \Z.
\end{equation}
The ACVF for univariate processes is symmetric in $k$, $\gamma_y(k) = \gamma_y(-k)$. Let $\rho(k) = \gamma(k) / \gamma(0)$ be the autocorrelation function (ACF). A large $\rho(k)$ means that the process $k$ time steps ago is highly correlated with the present $y_t$. The sample ACFs $\widehat{\rho}(k)$ in Fig.\ \ref{fig:real_world_examples} show that, e.g., S\&P 500 daily returns are uncorrelated with their own past (stock market efficiency); yearly tree ring growth is highly correlated over time with significant lags even for $k \geq 100$ years; and intuitively temperature in month $t$ is  highly correlated with the temperature $k = 6$ (cold $\leftrightarrow$ warm) and $k = 12$ (cold $\rightarrow$ cold; warm $\rightarrow$ warm) months ago (or in the future).

The building block of time series models is \emph{white noise} $\varepsilon_t$, which has zero mean, finite variance, and is uncorrelated over time: $\varepsilon_t \sim WN(0, \sigma_{\varepsilon}^2)$ iff\footnote{\emph{Iff} will be used as an abbreviation for \emph{if and only if}.} 
\begin{inparaenum}[i)]
\item $\E \varepsilon_t = 0$,
\item  $\V \varepsilon_t = \gamma_{\varepsilon}(0) = \sigma_{\varepsilon}^2$, and
\item $\gamma_{\varepsilon}(k) = 0$ if $k \neq 0$.
\end{inparaenum}
Only if $\varepsilon_t$ is a Gaussian process, then it is also independent. 

For multivariate second-order stationary $\mathbf{X}_t$ with mean\footnote{Without loss of generality (WLOG) assume $\boldsymbol \mu = \boldsymbol 0$.} $\boldsymbol \mu \in \R^{n}$ and covariance matrix $\Sigma_{\mathbf{X}}$ the ACVF
\begin{equation}
\label{eq:ACVF_X_t}
\R^{n \times n} \ni \boldsymbol \Gamma_{\mathbf{X}}(k) = \E \left( \mathbf{X}_t - \boldsymbol \mu \right) \transpose{\left( \mathbf{X}_{t-k} - \boldsymbol \mu \right)},
\end{equation}
is a matrix-valued function of $k \in \Z$. In particular, $\boldsymbol \Gamma_{\mathbf{X}}(0) = \Sigma_{\mathbf{X}}$. The diagonal of $\boldsymbol \Gamma_{\mathbf{X}}(k)$ contains the ACVF of each $\mathbf{X}_{i}(t)$; the off-diagonal element $\boldsymbol \Gamma_{\mathbf{X}}(k)_{(i,j)}$ is the cross-covariance between the $i$th and $j$th series at \text{lag $k$:}
\begin{equation}
\label{eq:ACVF_X_t_ij}
\gamma_{ij}(k) =  \E \left( \mathbf{X}_{i,t} - \mu_i \right) \left( \mathbf{X}_{j,t-k} - \mu_j \right) \in \R.
\end{equation}
 Contrary to $\gamma_y(k)$, $\boldsymbol \Gamma_{\mathbf{X}}(k)$ is not symmetric, but
\begin{equation}
\label{eq:ACVF_transpose_symmetric}
\boldsymbol \Gamma_{\mathbf{X}}(k) = \transpose{\boldsymbol \Gamma_{\mathbf{X}}(-k)}.
\end{equation}

\subsection{Spectrum and Spectral Density}
The spectrum of a univariate stationary process can be defined as the Fourier transform of its ACVF,
\begin{align}
\label{eq:spectrum_univariate}
S_y(\lambda) = \frac{1}{2 \pi} \sum_{j = -\infty}^{\infty} \gamma_y(j) e^{i j \lambda}, \quad \lambda \in [-\pi, \pi],
\end{align}
where $i = \sqrt{-1}$ is the imaginary unit. Since $\gamma_y(k)$ is symmetric, the spectrum is a real-valued, non-negative function, $S_y: [-\pi, \pi] \rightarrow \R^{+}$. For white noise $\varepsilon_t$ all $\gamma_{\varepsilon}(k) = 0$ if $k \neq 0$, thus $S_{\varepsilon}(\lambda) = \frac{\sigma_{\varepsilon}^2}{2 \pi}$ is constant for all $\lambda \in [-\pi, \pi]$. When $\gamma(k) > 0$ for $k \neq 0$ the spectrum has peaks at the corresponding frequencies. For example, the spectral density of monthly temperature series (right in Fig.\ \ref{fig:real_world_examples}) has large peaks at $\lambda \approx \pi /6$ and $\pi/12$, which represent the half- and one-year cycle.\footnote{Frequencies $\lambda$ are often scaled by $\pi$, $\tilde{\lambda} = \lambda / \pi$. This does not change results qualitatively, but simplifies interpretation since the corresponding cycle length equals $\tilde{\lambda}^{-1}$.}

Vice versa, the ACVF can be recovered from the spectrum using the inverse Fourier transform
\begin{align}
\label{eq:inverse_Fourier}
\gamma_y(k) = \int_{-\pi}^{\pi} S_y(\lambda) e^{-i k \lambda} d \lambda, \quad k \in \Z.
\end{align}
In particular, $\int_{-\pi}^{\pi} S_y(\lambda) d \lambda = \sigma^2_y$ for $k = 0$. Let
\begin{equation}
\label{eq:spectral_density}
f_y(\lambda) = \frac{S_y(\lambda)}{\sigma_y^2} = \frac{1}{2 \pi} \sum_{j = -\infty}^{\infty} \rho_y(j) e^{i j \lambda},
\end{equation}
be the \emph{spectral density} of $y_t$. As $f_y(\lambda) \geq 0$ and $\int_{-\pi}^{\pi} f_y(\lambda) d \lambda = 1$, the spectral density can be interpreted as a probability density function (pdf) of an (unobserved) random variable (RV) $\Lambda$ that ``lives'' on the unit circle. For white noise $f_{\varepsilon}(\lambda) = \frac{1}{2 \pi}$, which represents the uniform distribution $U(-\pi, \pi)$.

\begin{remark}[Spectrum and spectral density]
In the time series literature ``spectrum'' and ``spectral density'' are often used interchangeably.  Here I reserve ``spectral \emph{density}'' for $f_y(\lambda)$ in \eqref{eq:spectral_density}, as it integrates to one such as standard probability \emph{density} functions.
\end{remark}

\section{Measuring Forecastability}
\label{sec:measuring_forecastability}
Forecasting is inherently tied to the time domain. Yet, since Eqs.\ \eqref{eq:spectrum_univariate} \& \eqref{eq:inverse_Fourier} provide a one-to-one mapping between the time and frequency domain, we can use frequency domain properties to measure forecastability.

The intuition for the proposed measure of forecastability is as follows. Consider 
\begin{align}
\label{eq:motivation_process}
\begin{split}
y_t &= \sqrt{2} \cos \left( 2\pi \mathcal{Y} t + \theta \right), \\
\theta & \sim U(-\pi, \pi), \quad \mathcal{Y} \sim p_y(y) \text { independent of } \theta. 
\end{split}
\end{align}
One can show that $S_y(\lambda) = p_y(\lambda)$ \citep{Gibson94_spectralentropy_interpretation}.

If we have to predict the future of $y_t$, then uncertainty about $y_{t+h}$, $h > 0$, is only manifested in uncertainty about $\mathcal{Y}$, since $\cos \left( 2\pi \mathcal{Y} t + \theta \right)$ is a deterministic function of $t$: less uncertainty about $\mathcal{Y}$ means less uncertainty about $y_{t+h}$.  We can measure this uncertainty using the Shannon entropy of $p_y(y)$ \citep{Shannon48}. It is thus natural to measure uncertainty about the future as (differential) entropy of $f_y(\lambda)$,
\begin{equation}
\label{eq:spectral_entropy}
H_{s,a}(y_t) := - \int_{-\pi}^{\pi} f_y(\lambda) \log_a f_y(\lambda) d\lambda,
\end{equation}
where $a > 0$ is the logarithm base.

On a finite support $[b, c]$ the maximum entropy occurs for the uniform distribution $U(b, c)$; thus a flat spectrum should indicate the least predictable sequence. And indeed, a flat spectrum corresponds to white noise, which is unpredictable by definition (using linear predictors). Consequently, for any stationary $y_t$
\begin{align*}
H_{s,a}(y_t) & \leq H_{s,a}(\text{white noise}) \\
&= - \int_{-\pi}^{\pi} \frac{1}{2 \pi} \log_a \frac{1}{2 \pi} d\lambda = \log_a 2 \pi,
\end{align*}
with equality iff $y_t$ is white noise.

\begin{definition}[Forecastability of a stationary process]
\label{def:Omega}
For a second-order stationary process $y_t$, let
\begin{align}
\label{eq:forecastability}
\begin{split}
\Omega & : y_t \mapsto [0, \infty], \\
 \Omega(y_t) & = 1 - \frac{H_{s,a}(y_t)}{\log_a \left( 2 \pi \right)} = 1 - H_{s,2 \pi}(y_t),
\end{split}
\end{align}
be the \emph{forecastability} of $y_t$.
\end{definition}

Contrary to other measures in the signal processing and time series literature, $\Omega(y_t)$ does not require actual forecasts, but is a characteristic of the process $y_t$. It is therefore not biased to a particular -- perhaps sub-optimal -- model, forecast horizon, or loss function; as used in e.g., \citet{Stone01_BSS_TemporalPredictability, BoxTiao77_CCAtimeseries}.

\begin{properties}
\label{prop:Omega}
 $\Omega(y_t)$ satisfies:
\begin{enumerate}[a)]
\item $\Omega(y_t) = 0$ iff $y_t$ is white noise.
\item \label{item:shift_scale_invariant} invariant to scaling and shifting:
\begin{equation*}
\Omega(a y_t + b) = \Omega(y_t) \text{ for } a,b \in \R, a \neq 0.
\end{equation*}
\item \label{item:Omega_comb_less_max} max sub-additivity for uncorrelated processes: 
\begin{align}
\label{eq:Omega_comb_less_max}
\Omega(\alpha x_t + \sqrt{1-\alpha^2} y_t) \leq  \max \lbrace \Omega(x_t), \Omega(y_t) \rbrace,
\end{align}
if $\E x_t y_s = 0 \text { for all } s, t \in \Z$; equality iff $\alpha \in \lbrace 0, 1 \rbrace$.
\end{enumerate}
\end{properties}

The three series in Fig.\ \ref{fig:real_world_examples} are ordered (left to right) by increasing forecastability and indeed larger $\widehat{\Omega}$ correspond to intuitively more predictable real-world events: stock returns are in general not predictable; average monthly temperature is. 

We can thus use \eqref{eq:forecastability} to guide the search for optimal $\mathbf{w}$ that make $y_t = \transpose{\mathbf{w}} \mathbf{X}_t$ as forecastable as possible.

\subsection{Plug-in Estimator for $\Omega$}
To estimate $\Omega(y_t)$, we first estimate $S_y(\lambda)$, normalize it, and then plug it in \eqref{eq:spectral_entropy}. 

An unbiased estimator of $S_y(\lambda)$ is the periodogram
\begin{equation}
\label{eq:periodogram}
I_{T, \mathbf{y}_1^T}(\omega_j) = \Big| \frac{1}{\sqrt{T}} \sum_{t=0}^{T-1} y_t e^{-2 \pi i \omega_j t}\Big|^2,
\end{equation}
where $\omega_j = j/T$, $j=0,1,\ldots, T-1$ are the (scaled) Fourier frequencies, and $\mathbf{y}_1^T = \lbrace y_1, \ldots, y_T \rbrace$ is a sample of $y_t$ . It is well known that \eqref{eq:periodogram} is not a good estimate (e.g., periodograms are not consistent).  In the numerical examples we therefore use weighted overlapping segment averaging (WOSA) \citep{NuttalCarter82_WOSA} $\widehat{S}_y(\omega_j)$ 
from the R package \texttt{sapa}: \texttt{SDF(y, ''wosa'')}.

The bottom row of Figure \ref{fig:real_world_examples} shows the normalized $\widehat{f}_{j,y} = \frac{\widehat{S}_y(\omega_j)}{\sum_{j=0}^{T-1} \widehat{S}_y(\omega_j)}$ along with the plug-in estimate
\begin{equation}
\label{eq:estimate_Omega}
\widehat{\Omega}(\mathbf{y}_1^T) = 1 + \sum_{j = 0}^{T-1} \widehat{f}_{j, y} \cdot \log_{a = T} \left( \widehat{f}_{j, y} \right) .
\end{equation}

\begin{remark}
Typically, to estimate $\E g(X)$ for $X \sim p(x)$ (here: $g(X) = \log p(X)$) the sample average is solely over $g(x_j)$ without multiplicative $p(x_j)$ terms. This however assumes that each $x_j$ is sampled from $p(x)$ (and thus $\frac{1}{n} \sum_{i=1}^{n} g(x_i) \rightarrow \E_p g(X) = \int g(x) p(x) dx$ by the strong law of large numbers). While this is true in a standard sampling framework, here the ``data'' are the Fourier frequencies $\omega_j$ and the fast Fourier transform (FFT) samples them uniformly (and deterministically) from $[-\pi, \pi]$ and not according to the ``true'' spectral density $f(\lambda)$.\footnote{Advances in ``compressed sensing'' \citep{Jacques_Vandergheynst_2010_CompressedSensing} might improve estimates; see also ``non-uniform FFT'' \citep{Fessler03nonuniformfast}.}
\end{remark}

Eq.\ \eqref{eq:estimate_Omega} can be improved by a better spectral density \citep{Fryzlewiczetal08_WaveletSpectrumEstimation, Trobs_Heinzel_2006, Lees_Park_1995} and entropy estimation \citep{Paninski03_EntropyEstimation}. Future research can also address direct estimation of \eqref{eq:spectral_entropy} -- as is common for classic entropy estimates \citep{SricharanRaichHero11_knn_entropy, Mark09_Fastmultidimensionalentropy}. However, since neither spectrum nor entropy estimation are the primary focus of this work, we use standard estimators for $S_y(\lambda)$ and then the plug-in estimator of \eqref{eq:estimate_Omega}. 

It must be noted though that $\widehat{\Omega}(\mathbf{y}_1^T)$ in \eqref{eq:estimate_Omega} is based on discrete rather than differential entropy.  It still has the intuitive property that white noise has zero estimated forecastability, but now $\widehat{\Omega}(\mathbf{y}_1^T) \in [0, 1]$; $\widehat{\Omega}(\mathbf{y}_1^T) = 1$ iff the sample is a perfect sinusoid. Applications show that \eqref{eq:estimate_Omega} yields reasonable estimates and we do not expect the results to change qualitatively for other estimators.  We leave differential entropy estimates of $\Omega$ to future work.

Notice that $\Omega(y_t)$ relies on Gaussianity as only then $f_y(\lambda)$ captures all the temporal dependence structure of $y_t$. While time series are often non-Gaussian, $\Omega(\cdot)$ is a computationally and algebraically manageable forecastability measure  -- similarly to the importance of variance in PCA for iid data, even though they are rarely Gaussian.

\section{ForeCA: Maximizing Forecastability}
\label{sec:maximize_Omega}

Recall from Eq.\ \eqref{eq:max_problem} that we want to find a linear combination of a multivariate $\mathbf{X}_t$ that makes $y_t = \transpose{\mathbf{w}} \mathbf{X}_t$ as forecastable as possible.  Based on the forecastability measure in Section \ref{sec:measuring_forecastability}, we can now formally define the ForeCA optimization problem:
\begin{align}
\label{eq:maximize_Omega}
\max_{\mathbf{w}} \Omega(\transpose{\mathbf{w}} \mathbf{X}_t) & = \max_{\mathbf{w}} \left( 1 + \frac{ \int_{-\pi}^{\pi} f_{y}(\lambda) \log_a f_y(\lambda) d \lambda }{\log_a \left( 2 \pi \right)} \right),\\
\label{eq:maximize_Omega_unit_variance_restriction}
\text{subject to } & \transpose{\mathbf{w}} \Sigma_{\mathbf{X}} \mathbf{w} = 1,
\end{align}
where \eqref{eq:maximize_Omega_unit_variance_restriction} must hold since \eqref{eq:spectral_entropy} uses the spectral \emph{density} of $y_t$, i.e.\ we need $\V y_t = \transpose{\mathbf{w}} \Sigma_{\mathbf{X}} \mathbf{w} = 1$.

Property \ref{prop:Omega}\ref{item:Omega_comb_less_max} seems to let \eqref{eq:maximize_Omega} only have a trivial boundary solution. However, it is intuitively clear that combining uncorrelated series makes forecasting (in general) more difficult, e.g., signal $+$ noise. But if $\E x_t y_s \neq 0$ for some $s,t \in \Z$ then combining them can make it simpler: for some $\alpha \in (0,1)$ it holds $\Omega(\alpha x_t + \sqrt{1-\alpha^2} y_t) > \max \lbrace \Omega(x_t), \Omega(y_t) \rbrace $.

To optimize the right hand side of \eqref{eq:maximize_Omega} we need to evaluate $f_y(\lambda) = f_{\transpose{\mathbf{w}} \mathbf{X}_t}(\lambda)$ for various $\mathbf{w}$ and do this efficiently.  We now show how to obtain $f_y(\lambda)$ by simple matrix-vector multiplication from $f_{\mathbf{X}}(\lambda)$.

\subsection{Spectrum of Multivariate Time Series and Their Linear Combinations}

For multivariate $\mathbf{X}_t$ the spectrum equals 
\begin{equation}
\label{eq:spectrum_multivariate}
S_{\mathbf{X}}(\lambda) = \frac{1}{2 \pi} \sum_{k = -\infty}^{\infty} \boldsymbol \Gamma_{\mathbf{X}}(k) e^{2 \pi i k \lambda}, \quad \lambda \in [-\pi, \pi].
\end{equation}
Contrary to the univariate case, \eqref{eq:spectrum_multivariate} is in general complex-valued. Yet, since $\boldsymbol \Gamma_{\mathbf{X}}(k) = \transpose{\boldsymbol \Gamma_{\mathbf{X}}(-k)}$, $S_{\mathbf{X}}(\lambda) \in \C^{n \times n}$ is Hermitian for every $\lambda$, $S_{\mathbf{X}}(\lambda) = \overline{\transpose{S_{\mathbf{X}}(\lambda)}}$, where $\overline{z} = a - i b$ is the complex conjugate of $z = a+ ib \in \C$ \citep[][p.\ 436]{BrockwellDavis91}.

For dimension reduction we consider linear combinations $y_t = \transpose{\mathbf{w}} \mathbf{X}_t$, $\mathbf{w} \in \R^{n}$. By assumption $\E y_t = \transpose{\mathbf{w}} \E \mathbf{X}_t = 0$ and $\gamma_y(k) = \E y_t y_{t-k} = \transpose{\mathbf{w}} \boldsymbol \Gamma_{\mathbf{X}}(k) \mathbf{w}$. In particular, $\gamma_y(0) = \sigma_y^2 = \transpose{\mathbf{w}} \Sigma_{\mathbf{X}} \mathbf{w}$. The spectrum of $\transpose{\mathbf{w}} \mathbf{X}_t$ can be quickly computed via $S_y(\lambda) = \transpose{\mathbf{w}} S_{\mathbf{X}}(\lambda) \mathbf{w} $ and consequently
\begin{equation}
\label{eq:linear_combination_spectral_density}
f_y(\lambda) = \frac{\transpose{\mathbf{w}} S_{\mathbf{X}}(\lambda) \mathbf{w}}{\transpose{\mathbf{w}} \Sigma_{\mathbf{X}} \mathbf{w}}, \quad \lambda \in [-\pi, \pi].
\end{equation}

Since $f_y(\lambda) \geq 0$ for every $y_t$, $\transpose{\mathbf{w}} S_{\mathbf{X}}(\lambda) \mathbf{w} \geq 0$ for all $\mathbf{w} \in \R^{n}$; thus $S_{\mathbf{X}}(\lambda)$ is positive semi-definite.

\subsection{Solving the Optimization Problem}
\label{sec:solving_optim_problem}
Since $\Omega$ is invariant to shift and scale (Property \ref{prop:Omega}\ref{item:shift_scale_invariant}), we shall not only assume zero mean, but also contemporaneously uncorrelated observed signals with unit variance in each component. WLOG consider $\mathbf{U}_t = \Sigma_{\mathbf{X}}^{-1/2} \mathbf{X}_t$; thus $\E \mathbf{U}_t \transpose{\mathbf{U}_t} = \mathbf{I}_n$. Given $\widehat{\mathbf{W}}_U$ for $\mathbf{U}_t$, the transformation for $\mathbf{X}_t$ becomes $\widehat{\mathbf{W}}_X = \widehat{\mathbf{W}}_U \widehat{\Sigma}_X^{-1/2}$. Problem \eqref{eq:maximize_Omega} is then equivalent to
\begin{align}
\label{eq:minimize_neg_entropy_2}
\mathbf{w}^* &= \argmin_{\mathbf{w}, \vnorm{\mathbf{w}}_2 = 1} h(\mathbf{w}) 
\end{align}
where
\begin{equation}
\label{eq:h_w}
h(\mathbf{w}) = - \int_{-\pi}^{\pi} \transpose{\mathbf{w}} S_U(\lambda) \mathbf{w} \cdot \ell \left( \mathbf{w}; \lambda \right)  d \lambda,
\end{equation}
is the spectral entropy (Eq.\ \eqref{eq:spectral_entropy}) of $\transpose{\mathbf{w}} \mathbf{X}_t$ as a function of $\mathbf{w}$. We use $\ell\left(\mathbf{w}; \lambda \right) := \log \transpose{\mathbf{w}} S_U(\lambda) \mathbf{w} = \log f_{ \transpose{\mathbf{w}} \mathbf{U}}(\lambda)$ for better readability.

In practice we approximate \eqref{eq:h_w} with $\widehat{S}_U(\omega_j) \in \C^{n \times n}$ 
and thus obtain\footnote{We use \texttt{``wosa''} estimates (\texttt{sapa} R package).  However, any other estimate of $S_U(\lambda)$ can be used.}
\begin{align}
\label{eq:maximize_entropy_discrete}
\mathbf{w}^* &=\argmin_{\mathbf{w}, \vnorm{\mathbf{w}}_2 = 1} \widehat{h}_T(\mathbf{w}).
\end{align}
Here
\begin{equation}
\label{eq:h_w_T}
 \widehat{h}_T(\mathbf{w}) =  - \frac{1}{T} \sum_{j=1}^{T-1} \left( \transpose{\mathbf{w}} \widehat{S}_U(\omega_j) \mathbf{w} \right) \cdot \widehat{\ell} \left(\mathbf{w}; \omega_j \right)
\end{equation}
is the discretized version of \eqref{eq:minimize_neg_entropy_2}, where $\widehat{\ell} \left(\mathbf{w}; \omega_j \right) = \log \transpose{\mathbf{w}} \widehat{S}_U(\omega_j) \mathbf{w}$.  Notice that $\widehat{S}_U(\omega_j) \in \C^{n \times n}$ varies with $\omega_j$ while $\mathbf{w} \in \R^{n}$ is fixed over all frequencies, which makes it difficult to obtain an analytic, closed-form solution. However, \eqref{eq:maximize_entropy_discrete} can be solved iteratively borrowing ideas from the expectation maximization (EM) algorithm \citep{DempsterLairdRubin77_EMalgorithm}.

\subsubsection{A Convergent EM-like Algorithm}
\label{sec:EM}
For every $\mathbf{w} \in \R^{n}$, $\vnorm{\mathbf{w}}_2=1$, $h(\mathbf{w})$ has the form of a mixture model with weights $\widehat{\pi}(j \mid {\mathbf{w}}) := \transpose{\mathbf{w}} \widehat{S}_U(\omega_j) \mathbf{w} \geq 0$ and ``log-likelihood'' $\widehat{\ell} \left(\mathbf{w}; \omega_j \right)$. Since $\int_{-\pi}^{\pi}  f_{ \transpose{\mathbf{w}} \mathbf{U}}(\lambda) d \lambda = 1$, $\widehat{\pi}(j \mid {\mathbf{w}})$ is indeed a discrete probability distribution over $\lbrace \omega_j \mid 0=1, \ldots, T-1 \rbrace$.

Just as in an EM algorithm, the objective $h(\mathbf{w})$ can be optimized iteratively by first fixing $\mathbf{w} \leftarrow \mathbf{w}^{(i)}$ in $\widehat{\ell}(\mathbf{w}; \omega_j)$, and then minimizing the quadratic form
\begin{align}
\label{eq:update_w_last}
\mathbf{w}_{i+1} & =  \argmin_{\mathbf{w}, \vnorm{\mathbf{w}}_2 = 1}  \transpose{\mathbf{w}} \overline{\widehat{S}^{(i)}_U} \mathbf{w},
\end{align}
where $ \overline{\widehat{S}^{(i)}_U} = - \frac{1}{T}  \sum_{j=0}^{T-1} \widehat{S}_U(\omega_j) \cdot \ell(\mathbf{w}_i; \omega_j) $.

\begin{proposition}
\label{prop:semi_def}
 $\overline{\widehat{S}^{(i)}_U} $ is positive semi-definite. 
\end{proposition}

Thus \eqref{eq:update_w_last} can be solved analytically by the last eigenvector of $\overline{\widehat{S}^{(i)}_U}$ -- automatically guaranteeing $\vnorm{\mathbf{w}}_2 = 1$. The procedure iterates until $\vnorm{\mathbf{w}_{i+1} - \mathbf{w}_{i}} < tol$ for some tolerance level $tol$. For initialization we sample $\mathbf{w}_0$ from an $n$-dimensional uniform hyper-cube, $U_n(-1,1)$, and normalize to $\mathbf{w}_0 = \mathbf{w}_0 / \sqrt{\sum_{j=1}^{n} w_{j,0}^2}$.

\newcommand{\figWidth}{0.24}
\begin{figure*}[!t]
        \centering
        \begin{subfigure}[t]{\figWidth\textwidth}
                \centering
                \includegraphics[width=\textwidth]{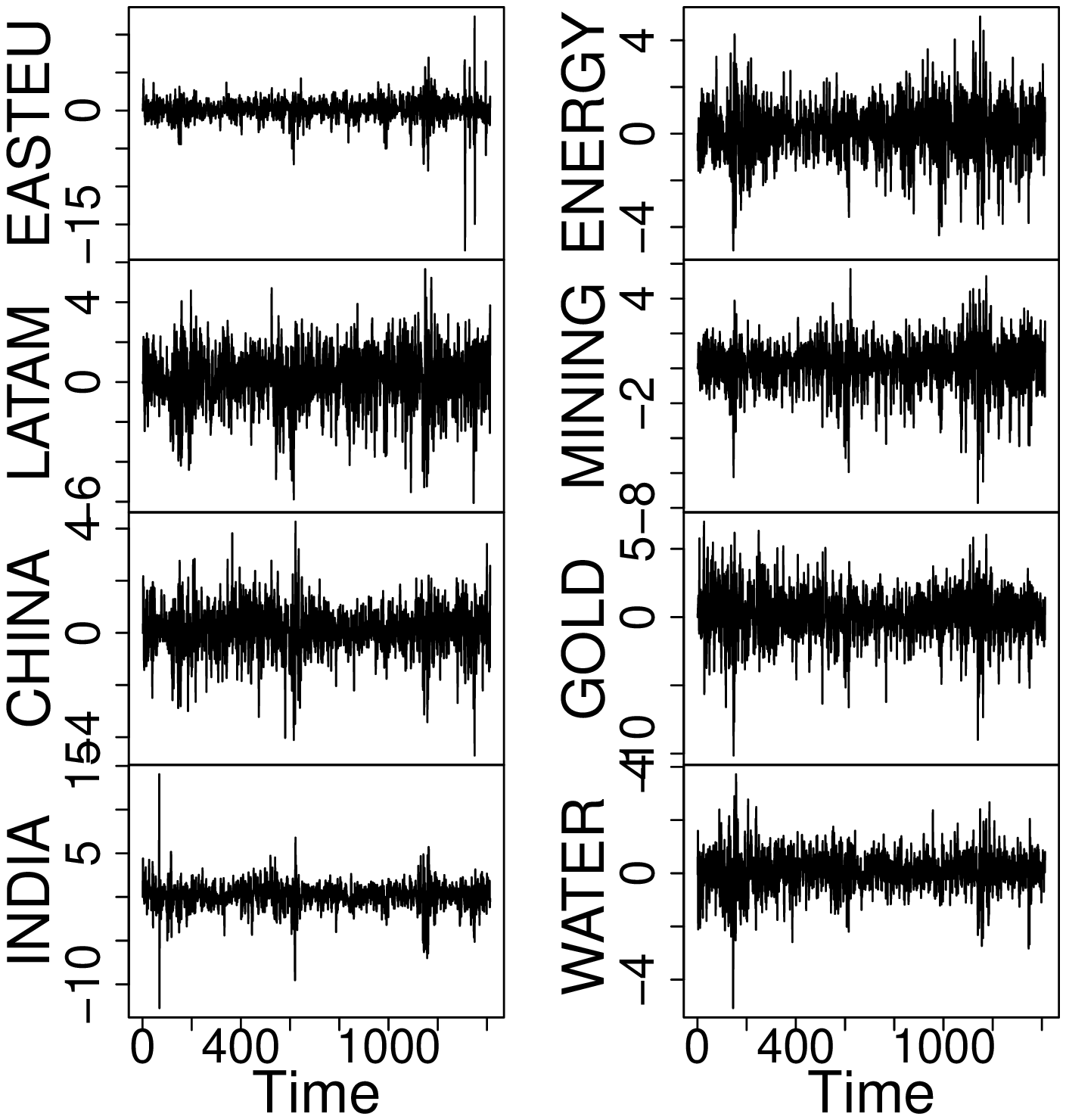}
                \caption{\label{fig:equityFunds_series} daily returns in $\%$}
        \end{subfigure}%
        \begin{subfigure}[t]{\figWidth\textwidth}
                \centering
                \includegraphics[width=\textwidth]{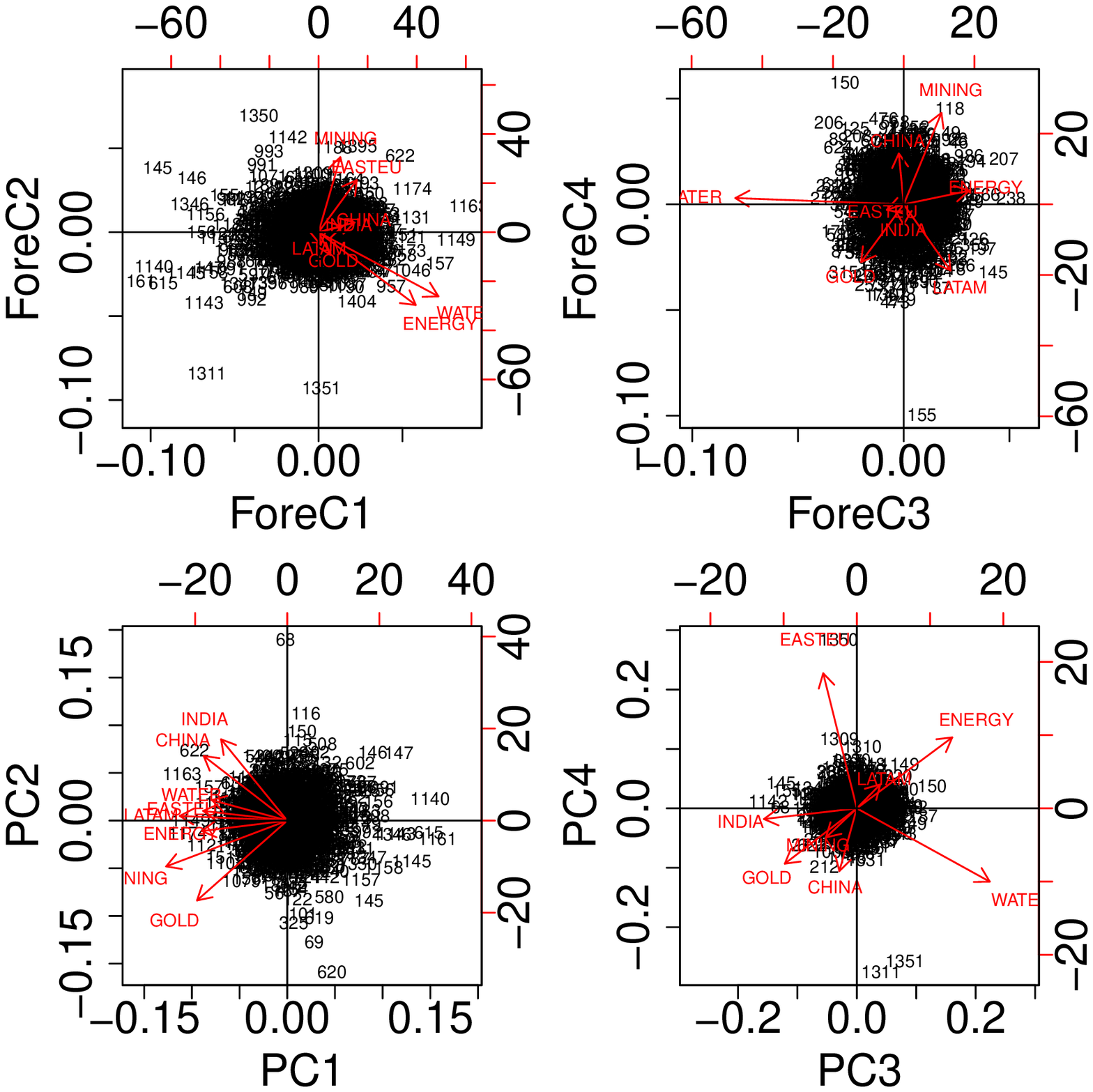}
                \caption{\label{fig:equityFunds_biplot_PCA_ForeCA} biplots of ForeCA (top) and PCA (bottom)}
        \end{subfigure}
        \begin{subfigure}[t]{\figWidth\textwidth}
                \centering
                \includegraphics[width=\textwidth]{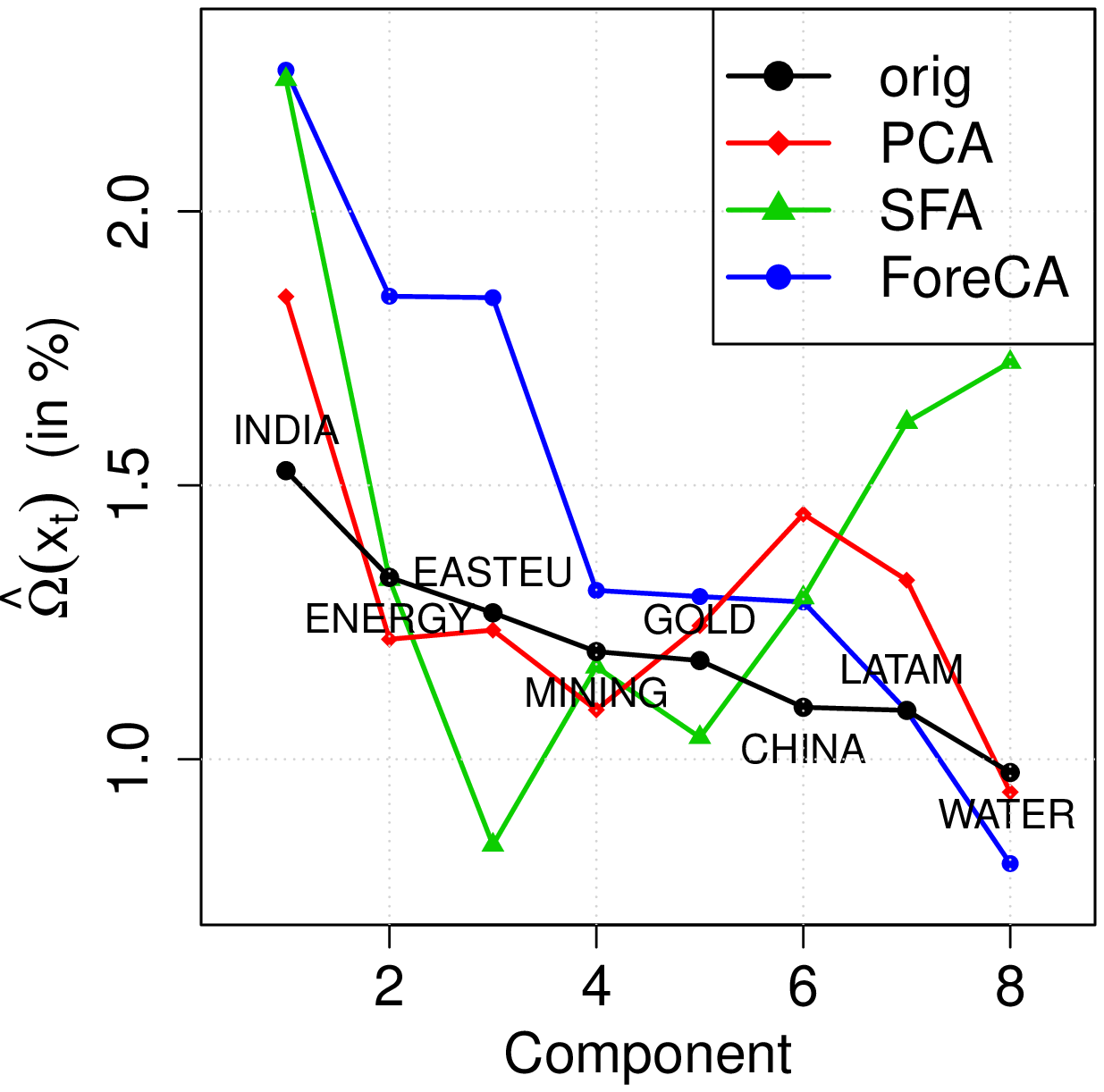}
                \caption{\label{fig:equityFunds_Omega_Orig_PCA_SFA_ForeCA} scree-plot of $\widehat{\Omega}(\cdot)$}
        \end{subfigure}%
        \begin{subfigure}[t]{\figWidth\textwidth}
                \centering
                \includegraphics[width=\textwidth]{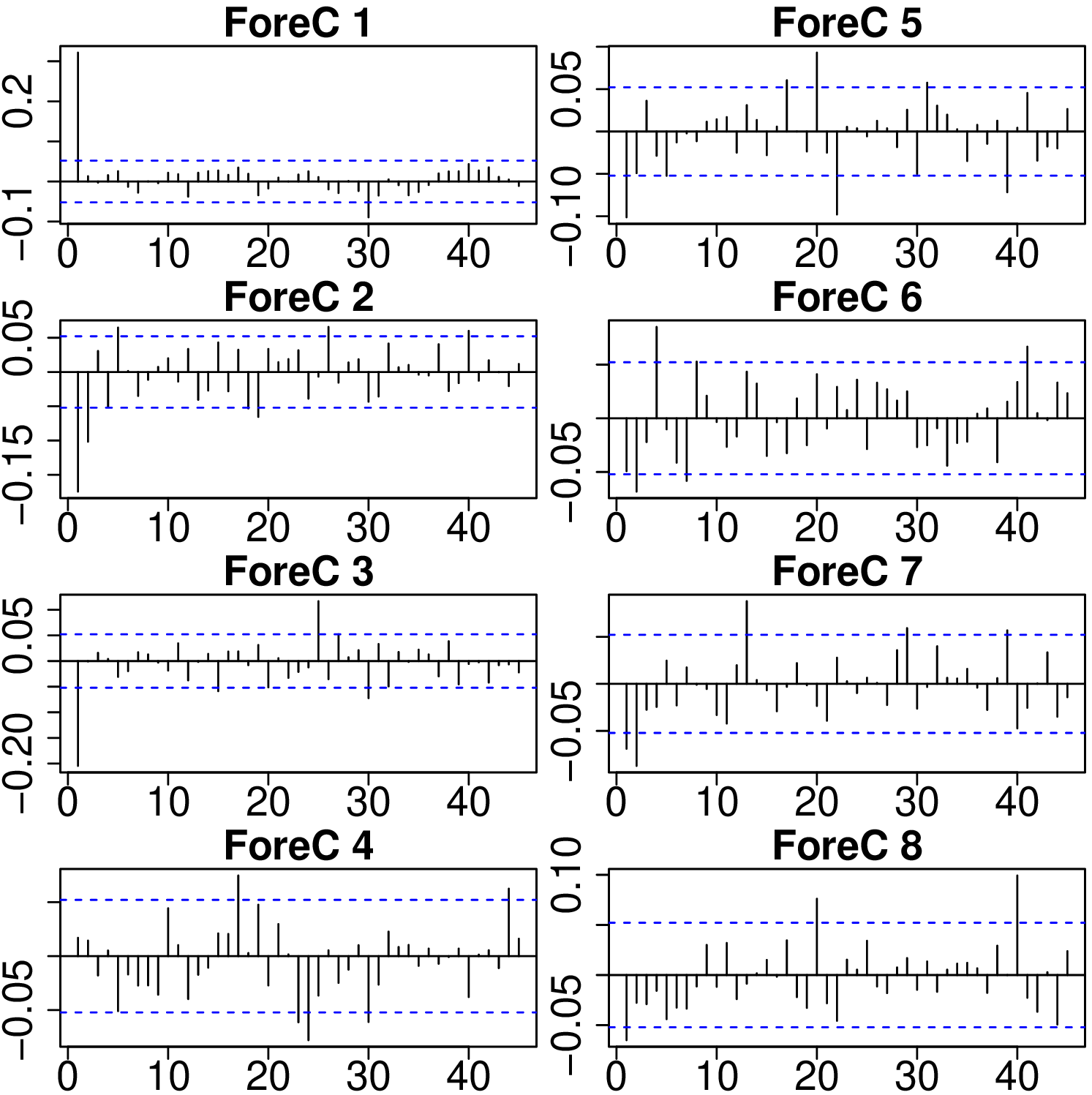}
                \caption{\label{fig:equityFunds_ForeCs_ACF} sample ACF $\widehat{\rho}(k)$ of ForeCs ($\widehat{\rho}(0) = 1$ omitted)}
        \end{subfigure}
        \caption[ForeCA results.]{\label{fig:ForeCA_equityFunds} Equity fund returns analyzed with PCA, SFA, and ForeCA. (Dataset \texttt{equityFunds} in R package \texttt{fEcofin}.)}
\end{figure*}

\begin{theorem}[Convergence]
\label{thm:convergence}
The sequence $\lbrace \mathbf{w}_i \rbrace_{i \geq 0}$ obtained via \eqref{eq:update_w_last} converges to a local minimum $\widehat{h}_T(\mathbf{w}^{*}) = \lambda_{\min}^{(*)} \geq 0$, where $\lim_{i \rightarrow \infty} \mathbf{w}_{i} = \mathbf{w}^*$ and $\lambda_{\min}^{(*)}$ is the smallest eigenvalue of $\overline{\widehat{S}^{(*)}_U}$.
\end{theorem}

\begin{corollary}
The transformed data $\mathbf{y}_{1}^{T, (*)} = \transpose{\mathbf{w}^{(*)} \,} \mathbf{X}_1^T$ satisfies
\begin{equation}
\widehat{\Omega}\left(\mathbf{y}_{1}^{T, (*)} \right) = 1 - \lambda_{\min}^{*}.
\end{equation}
\end{corollary}

\begin{proof}[Proof of Theorem \ref{thm:convergence}]
The entropy of a RV taking values in a finite alphabet $\lbrace \omega_0, \ldots, \omega_{T-1} \rbrace$ is bounded: $ 0 \leq \widehat{h}_T(\mathbf{w}) \leq \log_a T$ for all $\mathbf{w} \in \R^{n}$. For convergence it remains to be shown that $\widehat{h}_T(\mathbf{w}_{i}) \geq \widehat{h}_T(\mathbf{w}_{i+1})$ with equality iff $\mathbf{w}_{i+1} = \mathbf{w}_i = \mathbf{w}^*$. First, 
\begin{align}
\widehat{h}_T(\mathbf{w}_i) 
& = - \frac{1}{T} \sum_{j=1}^{T-1} \transpose{\mathbf{w}_i} S_U(\omega_j) \mathbf{w}_i \cdot \widehat{\ell}(\mathbf{w}_i; \omega_j)\nonumber \\
 \label{eq:eigenvector_increases_h}
& = \transpose{\mathbf{w}_i} \overline{\widehat{S}^{(i)}_U} \mathbf{w}_i  \geq \transpose{\mathbf{w}_{i+1}} \overline{\widehat{S}^{(i)}_U} \mathbf{w}_{i+1}
\end{align}
since $\mathbf{w}_{i+1}$ is the last eigenvector of $ \overline{\widehat{S}^{(i)}_U} $. Second,
\begin{align}
\transpose{\mathbf{w}_{i+1}} \overline{\widehat{S}^{(i)}_U} \mathbf{w}_{i+1}  & = - \frac{1}{T}\sum_{j=1}^{T-1} \transpose{\mathbf{w}_{i+1}} S_U(\omega_j) \mathbf{w}_{i+1} \cdot \widehat{\ell}(\mathbf{w}_i; \omega_j) \nonumber \\
\label{eq:KL_divergence_greater_0_increases_h}
& \geq - \frac{1}{T} \sum_{j=1}^{T-1} \transpose{\mathbf{w}_{i+1}} S_U(\omega_j) \mathbf{w}_{i+1} \cdot \widehat{\ell}(\mathbf{w}_{i+1}; \omega_j)\\
& = \widehat{h}_T(\mathbf{w}_{i+1}), \nonumber
\end{align}
where \eqref{eq:KL_divergence_greater_0_increases_h} holds as $\E_p -\log q = - \sum_{j=1}^{n} p_j \log q_j \geq - \sum_{j=1}^{n} p_j \log p_j = \E_p - \log p$ for any $q \neq p$.
\end{proof}

To lower the chance of landing in local optima we repeat \eqref{eq:update_w_last} for several random starting positions $\mathbf{w}_0$ and then select the best solution.

\subsection{Obtaining a $K$-dimensional Subspace}

To obtain all $K$ loadings $\mathbf{W}_{1, \ldots, K} = \left[ \mathbf{w}_1, \ldots, \mathbf{w}_K \right]$ that give uncorrelated series $y_{j, t}$, we iteratively (starting at $k=1$)
\begin{inparaenum}[i)]
\item compute $\mathbf{w}_k$,
\item project $\mathbf{U}$ onto the null space of $\mathbf{W}_{1, \ldots, k}$ $\rightarrow \mathbf{U}^{(k)} = \mathbf{W}_{1, \ldots, k}^{\perp} \mathbf{U} \in \R^{K-k}$,
\item apply the EM-type algorithm on $\mathbf{U}^{(k)}$ to obtain $\tilde{\mathbf{w}}_{k+1}$, and finally
\item transform $\tilde{\mathbf{w}}_{k+1}$ back to loadings $\mathbf{w}^{(k)}$ of $\mathbf{U}$.
\end{inparaenum}

Doing this for $k = 1, \ldots, K$ gives $K$ loadings $\widehat{\mathbf{W}}_U$.  Loadings for $\mathbf{X}_t$ are given by $\widehat{\mathbf{W}}_X = \widehat{\mathbf{W}}_U \widehat{\Sigma}_X^{-1/2}$.

\section{Applications}
\label{sec:applications}

Here we demonstrate the usefulness of ForeCA to find informative, forecastable signals, but also as a tool for time series classification.

\subsection{Improving Portfolio Forecasts}

Figure \ref{fig:equityFunds_series} shows daily returns of eight equity funds from $2002/01/01$ to $2007/05/31$ ($T = 1413$). In the financial context finding forecastable series is an important goal by itself, not just for structure discovery. In particular, we can interpret a linear combination $\mathbf{w}$ as a portfolio of stocks. The $\mathbf{w}^*$ with the highest $\Omega$ gives the most forecastable portfolio.

Figure \ref{fig:equityFunds_biplot_PCA_ForeCA} shows a bi-plot for PCA and ForeCA for $(\mathbf{w}_1, \mathbf{w}_2)$ and $(\mathbf{w}_3, \mathbf{w}_4)$. As PC $1$ weighs all funds almost equally, it represents the average market movement; the second component contrasts Gold \& Mining with the rest and we can therefore label PC $2$ as the ``commodity'' index. The third and fourth PC indicate energy/infrastructure and geographic regions. 

However, even though PC $1$ is also the most predictable PC, it has only a slightly larger $\widehat{\Omega}$ than the most forecastable fund, India (Fig.\ \ref{fig:equityFunds_Omega_Orig_PCA_SFA_ForeCA}). On the other hand, combining Water (weight $w_{water, 1} = 0.72$) with Energy ($0.58$) is almost twice as forecastable as India (weights are from ForeC $1$ in Fig.\ \ref{fig:equityFunds_biplot_PCA_ForeCA}).  ForeC $2$ also has high forecastability by selling Energy \& Water ($-0.53$ \& $-0.47$) and buying Mining \& Eastern Europe ($0.55$ \& $0.38$). The third and fourth ForeCs seem to be hedging strategies (ForeC $3$: Water vs.\ Energy; ForeC $4$: Latin America \& Gold vs.\ China \& Mining). 

As financial data only has very small autocorrelation -- and usually at lag $1$, if any --, SFA and ForeCA yield overall very similar results, except for a ``wrong'' ranking by SFA (Fig.\ \ref{fig:equityFunds_Omega_Orig_PCA_SFA_ForeCA}): SF $8$ is the fastest feature (large, but negative lag $1$ autocorrelation), yet it is the second most forecastable component.  While it is true that white noise is slower than an auto-regressive process of order $1$ ($AR(1)$) with negative autocorrelation, the latter is still more forecastable.  Since we want to reveal intertemporal structure, white noise must be ranked lowest; and ForeCA indeed does so (Fig.\ \ref{fig:equityFunds_ForeCs_ACF}).  

ForeC $5$ and $8$ detect the $20$ day lag (one trading month), but correlations are too low to achieve much higher forecastability than -- simpler and faster  -- SFA.
   
In the next example I study quarterly income data, where ForeCA can leverage its nonparametric power and detect important dependencies at various frequencies automatically from the data.

\subsection{Classification of US State Economies}
\label{sec:US_economies}

\renewcommand{\figWidth}{0.22}
\begin{figure*}[!t]
        \centering
        \begin{subfigure}[t]{\figWidth\textwidth}
                \centering
                \includegraphics[width=\textwidth, trim = 1.5cm 0cm 1.5cm 1.5cm, clip = true]{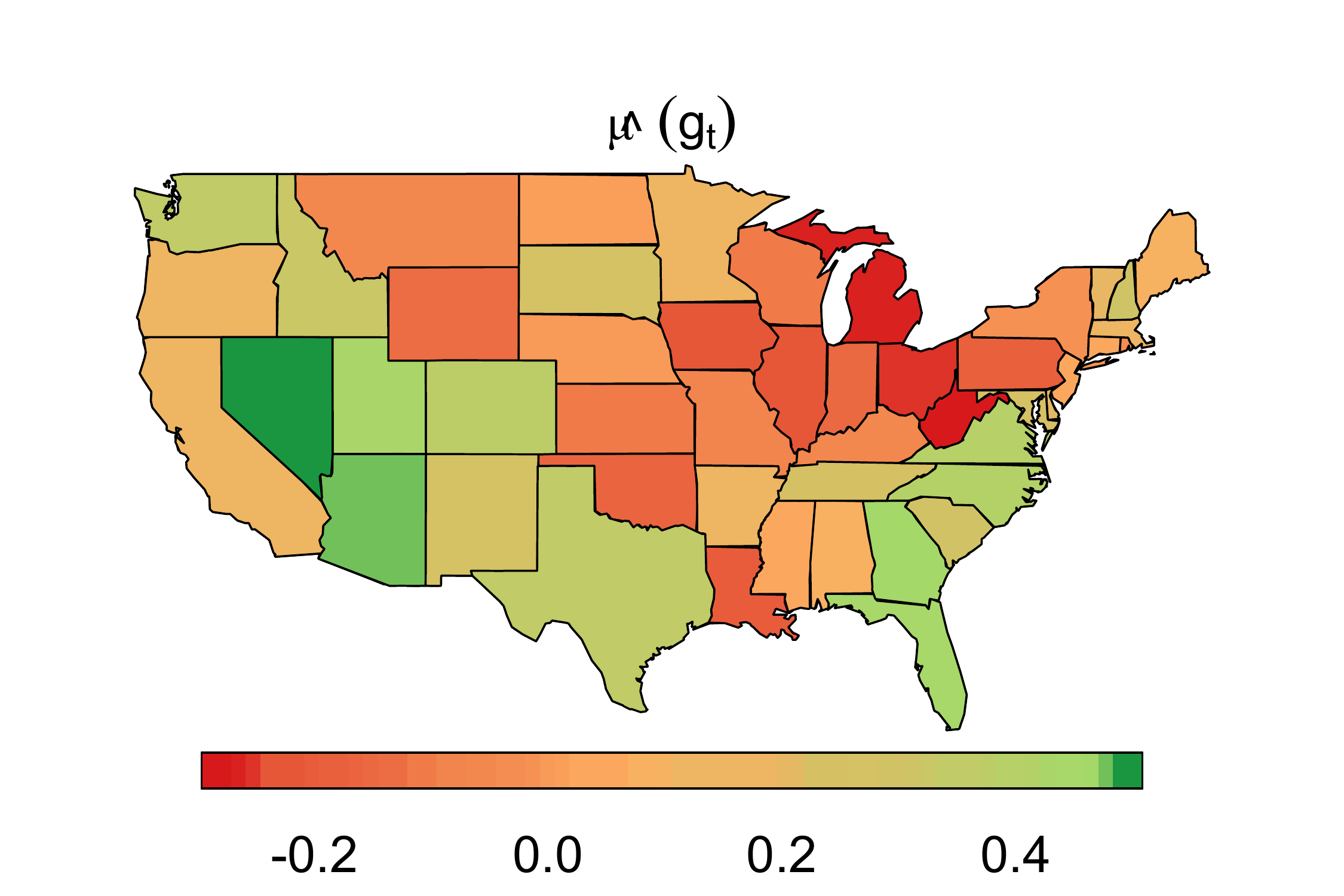}
                \caption{\label{fig:US_income_growth_map_mean} Average }
        \end{subfigure}
        \hspace{0.01\textwidth}
        \begin{subfigure}[t]{\figWidth\textwidth}
                \centering
                \includegraphics[width=\textwidth, trim = 1.5cm 0cm 1.5cm 1.5cm, clip = true]{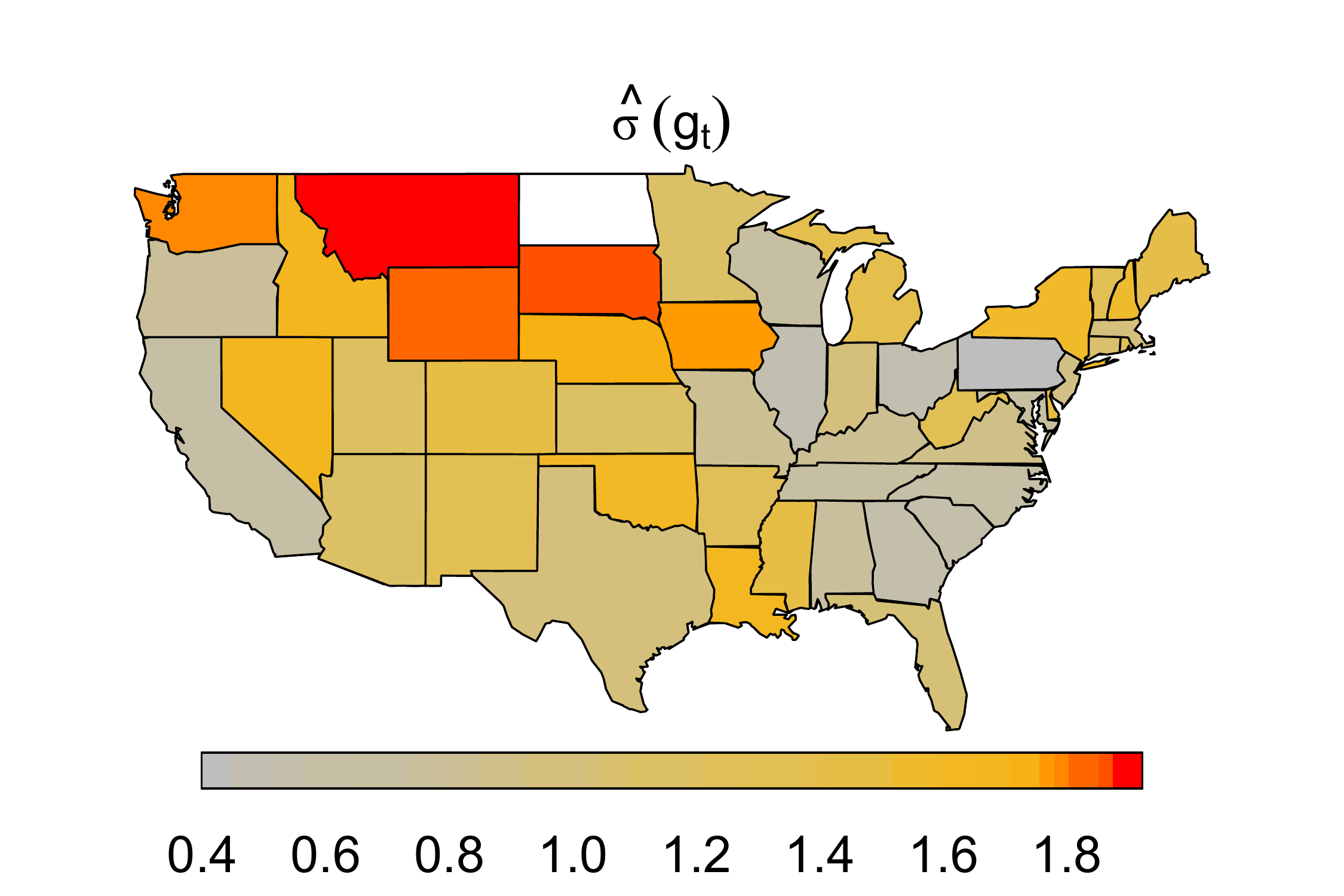}
                \caption{\label{fig:US_income_growth_map_sd} Standard deviation $\widehat{\sigma}$; ND omitted ($\widehat{\sigma}_{ND} = 2.98$).}
        \end{subfigure}%
        \hspace{0.01\textwidth}
        \begin{subfigure}[t]{\figWidth\textwidth}
                \centering
                \includegraphics[width=\textwidth, trim = 1.5cm 0cm 1.5cm 1.5cm, clip = true]{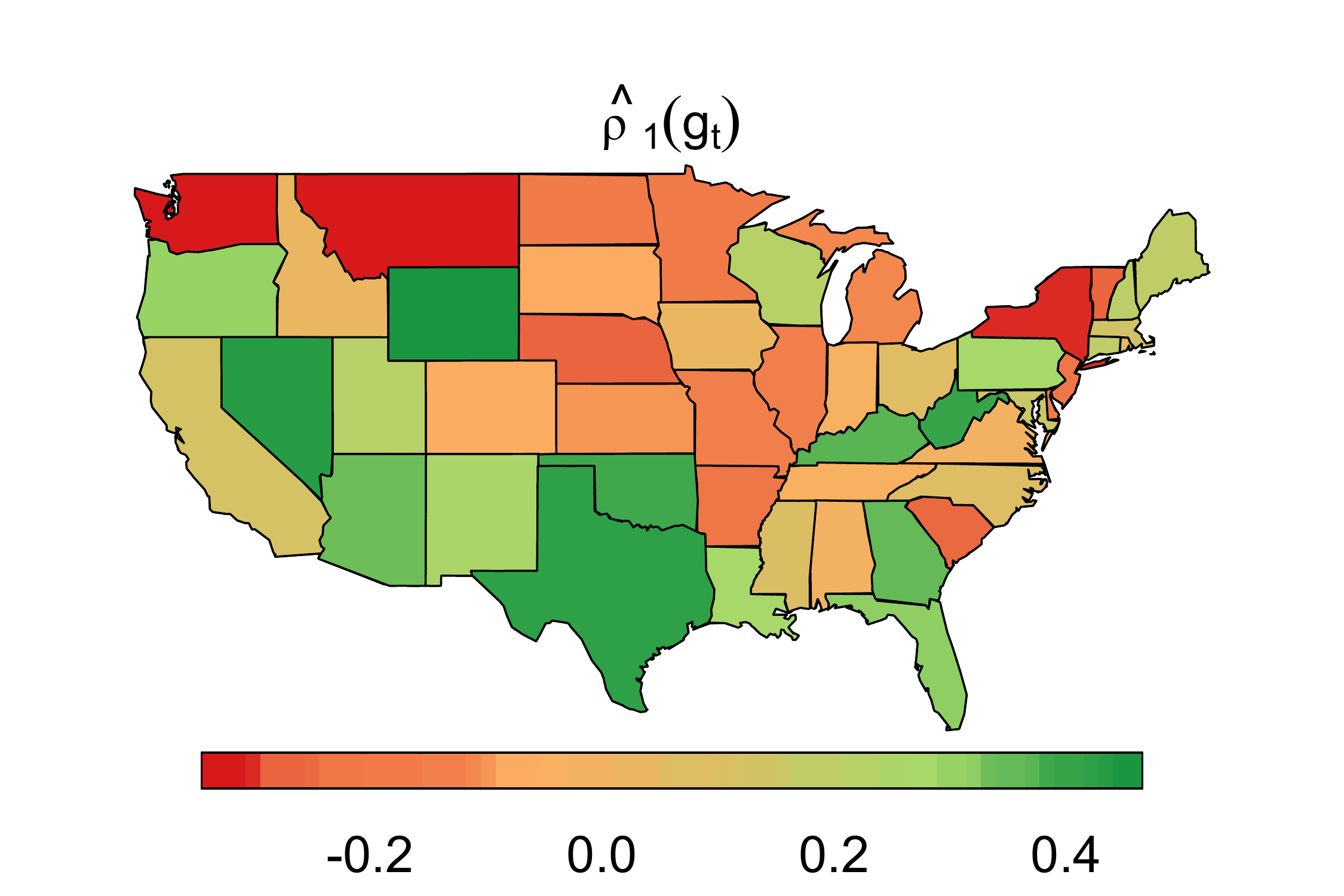}
                \caption{\label{fig:US_income_growth_map_rho1} Lag $k=1$ autocorrelation $\widehat{\rho}(1)$}
        \end{subfigure}
        \hspace{0.01\textwidth}
        \begin{subfigure}[t]{\figWidth\textwidth}
                \centering
                \includegraphics[width=\textwidth, trim = 1.5cm 0cm 1.5cm 1.5cm, clip = true]{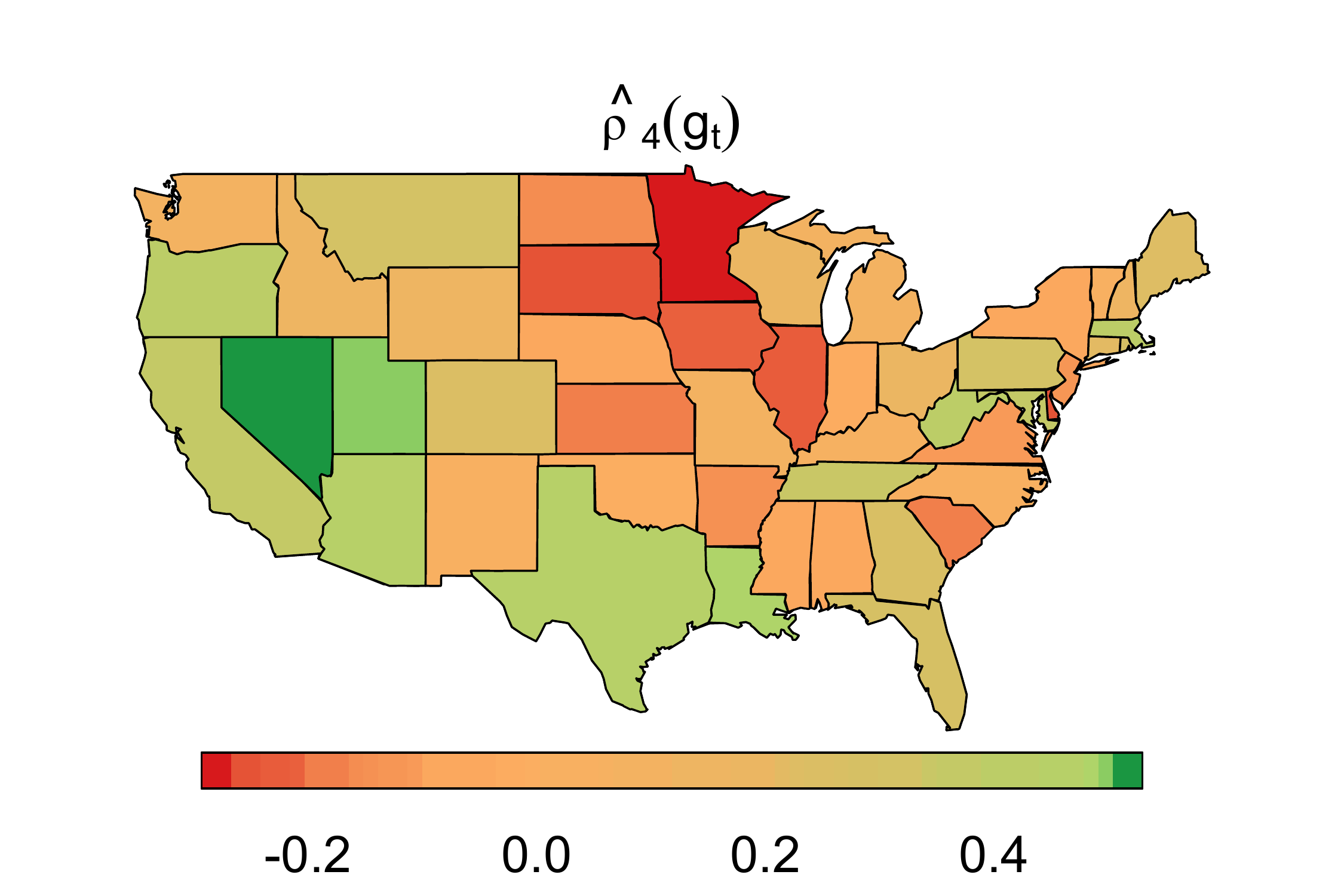}
                \caption{\label{fig:US_income_growth_map_rho4} Lag $k=4$ autocorrelation $\widehat{\rho}(4)$}
        \end{subfigure}%
        
\renewcommand{\figWidth}{0.22}
\begin{subfigure}[t]{\figWidth\textwidth}
                \centering
                \includegraphics[width=\textwidth, trim = 1.5cm 0cm 1.5cm 1.5cm, clip = true]{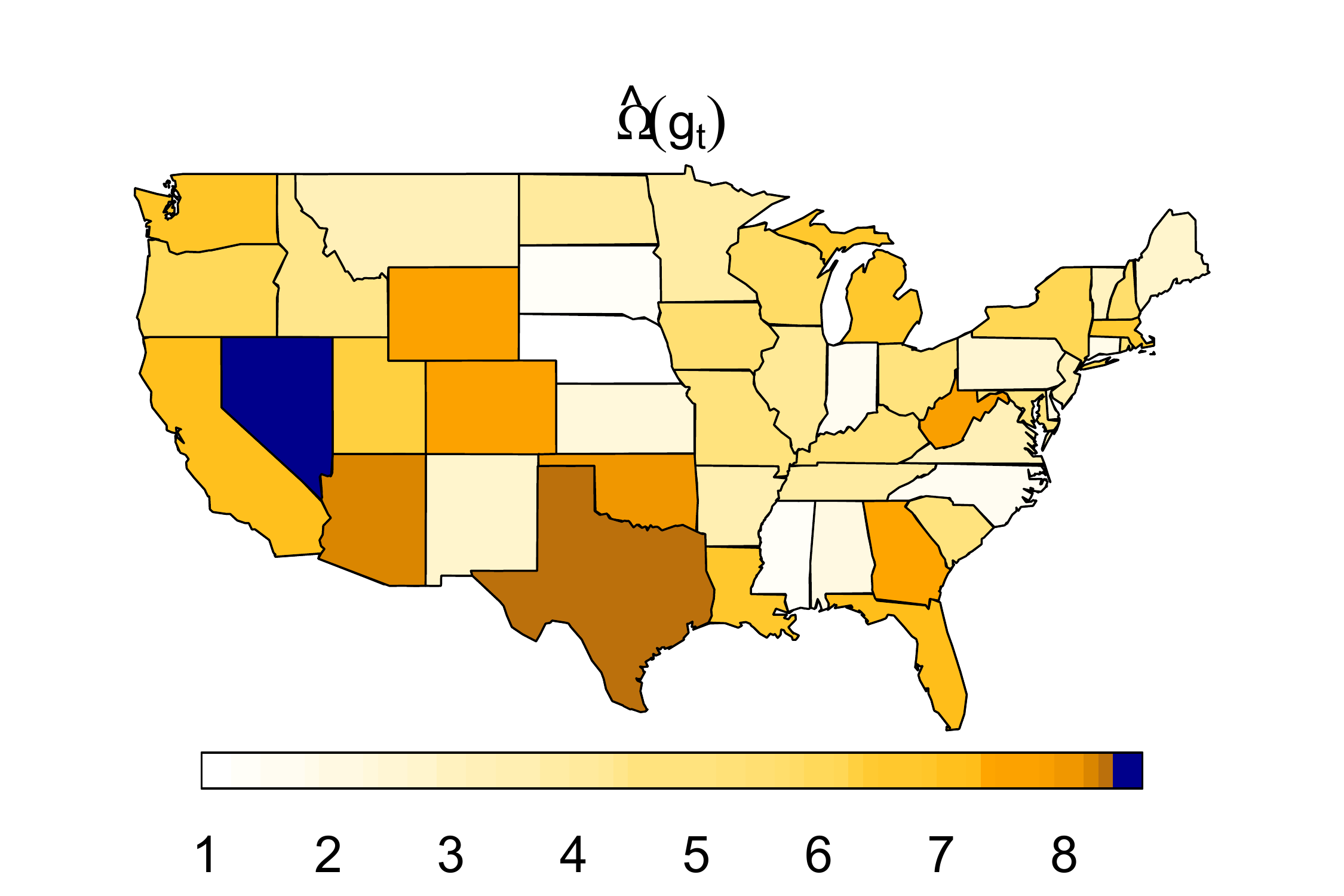}
                \caption{\label{fig:US_income_growth_map_Omega} Forecastability $\widehat{\Omega}$}
        \end{subfigure}
        \hspace{0.01\textwidth}
        \begin{subfigure}[t]{\figWidth\textwidth}
                \centering
                \includegraphics[width=\textwidth]{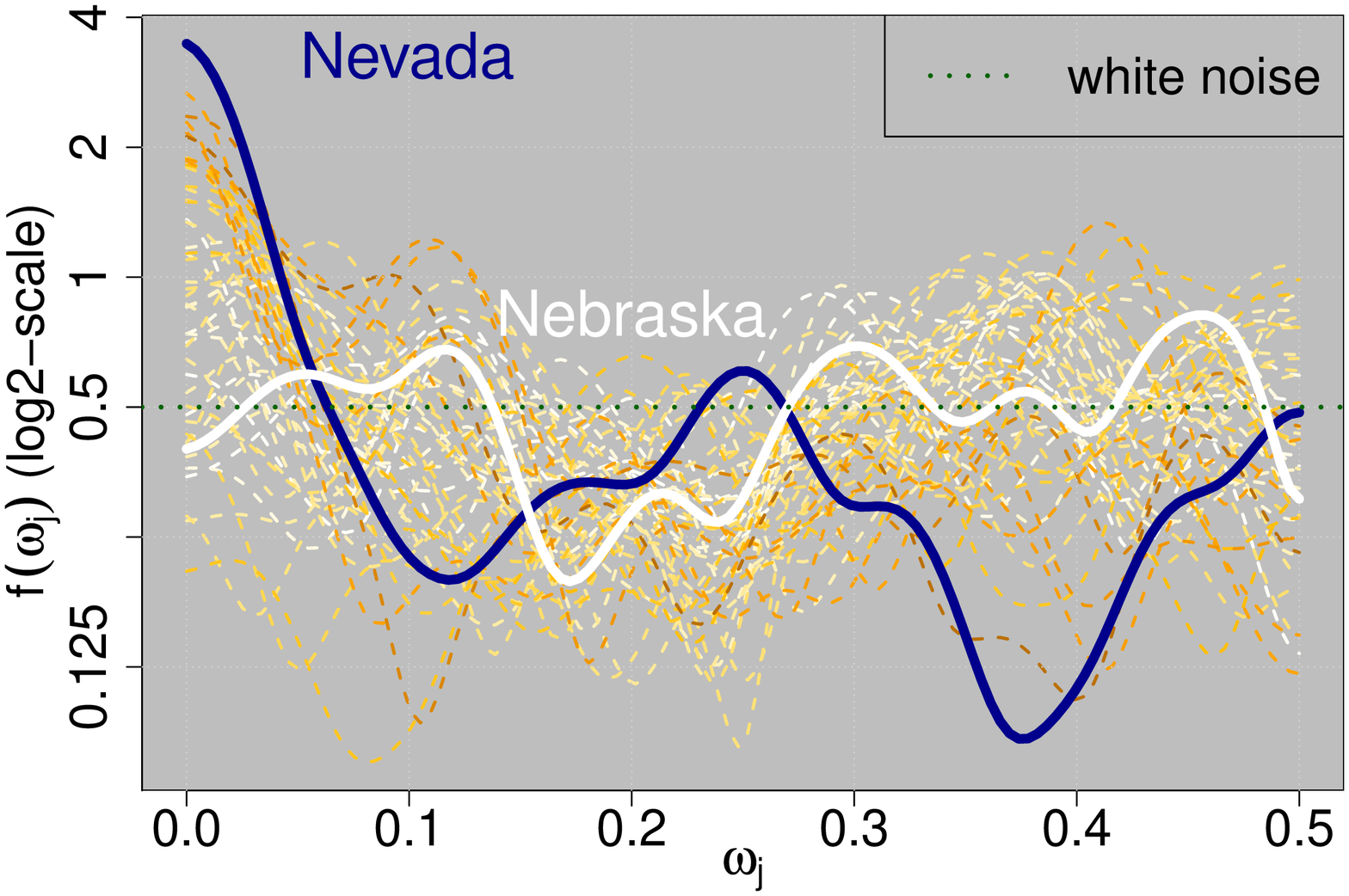}
                \caption{\label{fig:US_income_growth_spectra} Spectra $\widehat{f}_{g}(\lambda)$} 
        \end{subfigure}
        \hspace{0.01\textwidth}
        \begin{subfigure}[t]{\figWidth\textwidth}
                \centering
                \includegraphics[width=\textwidth, trim = 1.5cm 0cm 1.5cm 1.5cm, clip = true]{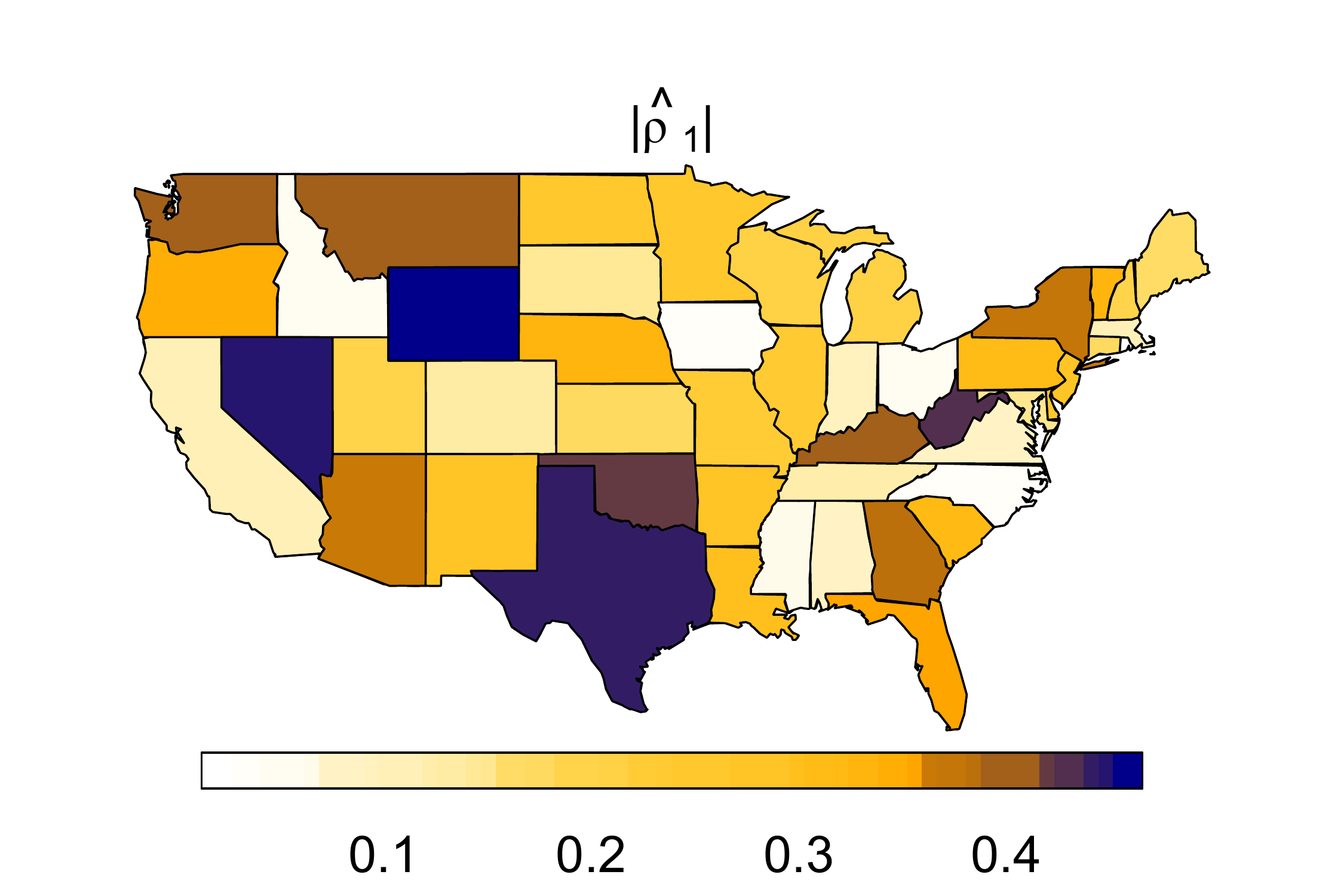}
                \caption{\label{fig:US_income_growth_map_absrho1} Absolute value of $\widehat{\rho}(1)$}
        \end{subfigure}%
        \hspace{0.01\textwidth}
        \begin{subfigure}[t]{\figWidth\textwidth}
                \centering
                \includegraphics[width=\textwidth, trim = 1.5cm 0cm 1.5cm 1.5cm, clip = true]{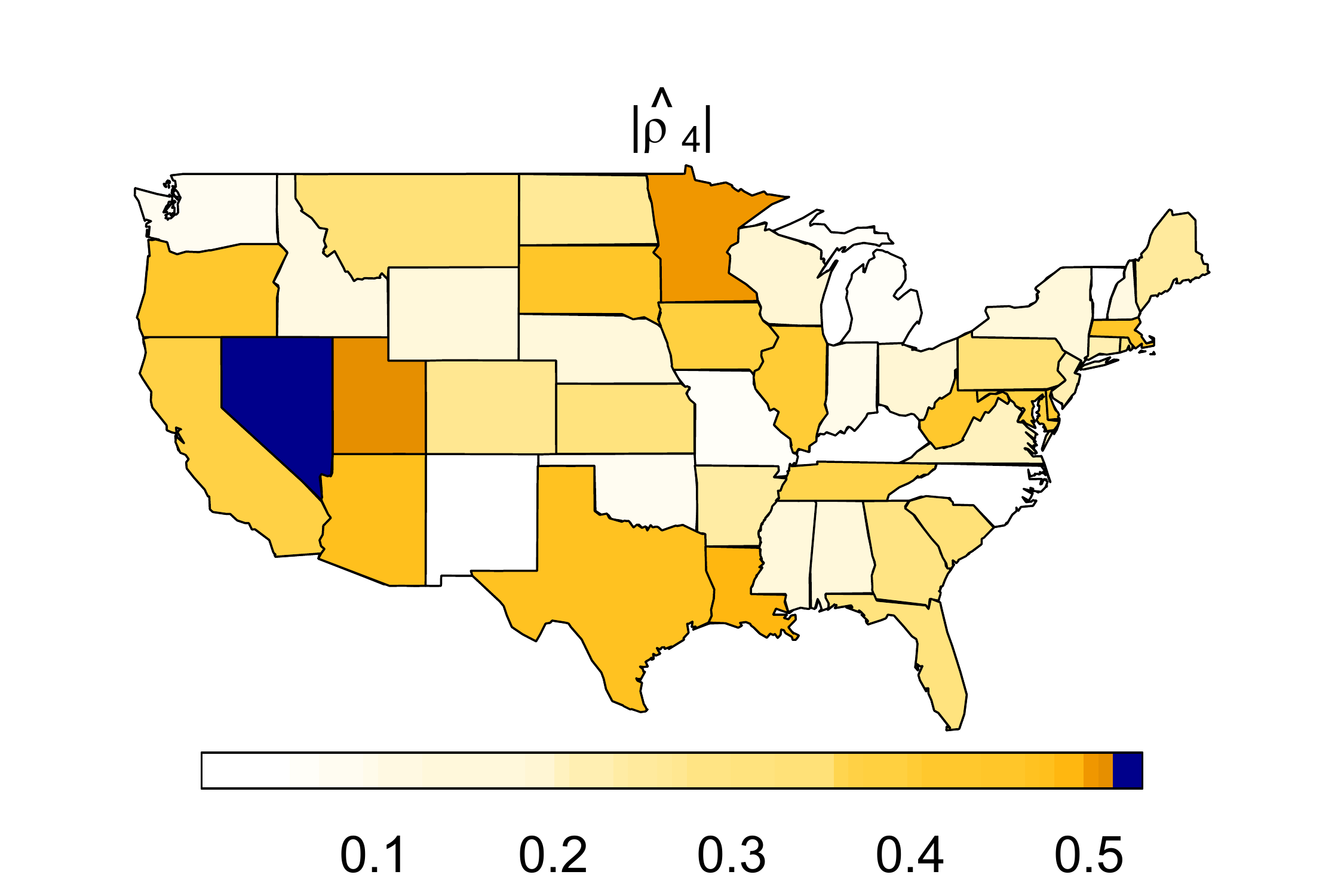}
                \caption{\label{fig:US_income_growth_map_absrho4} Absolute value of $\widehat{\rho}(4)$}
        \end{subfigure}%
        \caption[US States summary statistics.]{\label{fig:US_income_overview} Summary statistics of quarterly income growth rates (in \%) from $1982/1$ -- $2011/4$ with respect to US baseline $\widehat{\mu}(r_{US,t}) = 1.32\%$, $\widehat{\sigma}(r_{US,t}) = 0.92\% $ per quarter; $\widehat{\Omega}(r_{US,t}) = 4.86\% $, $\widehat{\rho}_1(r_{US,t}) = 0.42$, $\widehat{\rho}_4(r_{US,t}) = 0.13$.}
\end{figure*}

I consider quarterly per-capita income growth rates of the ``lower 48'' from $1982/1$ to $2011/4$ (last $30$ years)
\begin{equation*}
g_{j,t} = r_{j, t} - r_{US, t}, \quad j \in \lbrace \text{AL, $\ldots$, WY} \rbrace,
\end{equation*}
where $r_{j,t}$ is the annual growth rate of region $j$.\footnote{Publicly available at \url{www.bea.gov/itable}.}  Interested in finding similar state economies within the US, we subtract the US baseline. 
Clustering states with similar economic dynamics can help to decide where to provide support when facing difficult economic times. For example, if certain states do not show any important dynamics on a $7$-$8$ year scale -- also known as the ``business cycle'' \citep{Halletetal08_EuropeanBusinessCycle} -- then it might be better to support states that are affected by these global economy swings.

The first row of Fig.\ \ref{fig:US_income_overview} displays basic summary statistics: sample average, standard deviation, and first and fourth order autocorrelation.  The second row give statistics related to forecastability:  Fig.\ \ref{fig:US_income_growth_map_Omega} shows $\widehat{\Omega}$ based on the spectra in Fig.\ \ref{fig:US_income_growth_spectra}; Fig.\ \ref{fig:US_income_growth_map_absrho1} shows the absolute lag $1$ correlation (analogously for lag $4$ in Fig.\ \ref{fig:US_income_growth_map_absrho4}), since two $AR(1)$s with a $\pm \phi$ lag $1$ coefficient are equivalent in terms of forecasting (compare to SFA ranking in the portfolio example).  

The spectral densities of Nevada and  Nebraska illustrate the intuitive derivation of $\Omega(x_t)$ from Eq.\  \eqref{eq:motivation_process}: for Nebraska all frequencies are equally important and it is thus difficult to forecast any better than the sample mean; contrary, Nevada's income growth rates are mainly driven by a yearly cycle ($\omega_j \approx 0.25$) and low frequencies, thus Nevada is much easier to forecast.

A similar dataset (but annually and for different years) has been analyzed in \citet{Dhiral01_ClusteringARIMA}, who fit  $AR(1)$ models to the non-adjusted growth rates $r_{j,t}$ for $25$ pre-selected states, and then cluster them in the model space. Although they obtain interpretable results, it is unlikely that US state economies only differ in their lag $1$ coefficient. In particular, simple $AR(1)$ models cannot capture the business cycle, which is clearly visible in Fig.\ \ref{fig:US_income_growth_spectra} (even for the adjusted rates).

Similarly, as SFA maximizes lag $1$ correlation, it misses the quarterly cycle. ForeCA does not face this model selection bias, but can find forecastability across all frequencies. In particular, only ForeC $4$ detects interesting high frequency signals (Fig.\ \ref{fig:US_income_ForeCs}). The most forecastable PCs are PC $5$, $4$, and $1$; interestingly PC $3$ is least important for forecasting among all $48$ PCs.  Also note that ForeCs are more interpretable than SFs or PCs (Figs.\ \ref{fig:US_income_ForeCs} - \ref{fig:US_income_PCs}). Particularly, ForeC $1$ shows a clear $\approx 25$ year period (generation cycle), whereas PC $1$ looks somewhat arbitrary.  Yet, the associated loadings in Fig.\ \ref{fig:US_income_maps-pca-sfa-foreca} are quite similar.

\renewcommand{\figWidth}{0.48}
\begin{figure}[!t]
\centering
\begin{subfigure}[t]{\figWidth\textwidth}
        \centering
        \includegraphics[width=\textwidth]{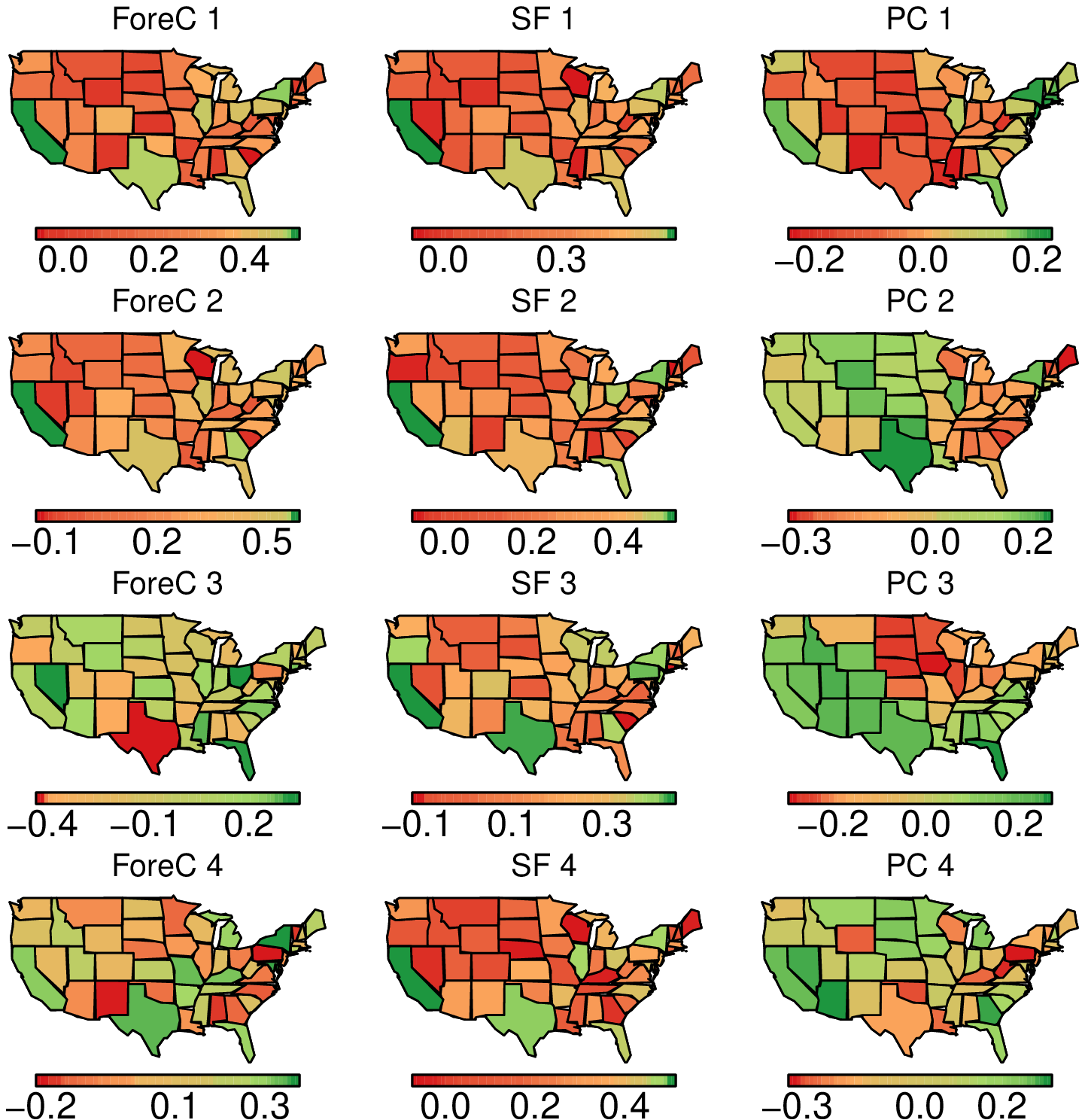}
        \caption{\label{fig:US_income_maps-pca-sfa-foreca} First $4$ loadings.}
\end{subfigure}%

\renewcommand{\figWidth}{0.15}
\begin{subfigure}[t]{\figWidth\textwidth}
        \centering
        \includegraphics[width=\textwidth]{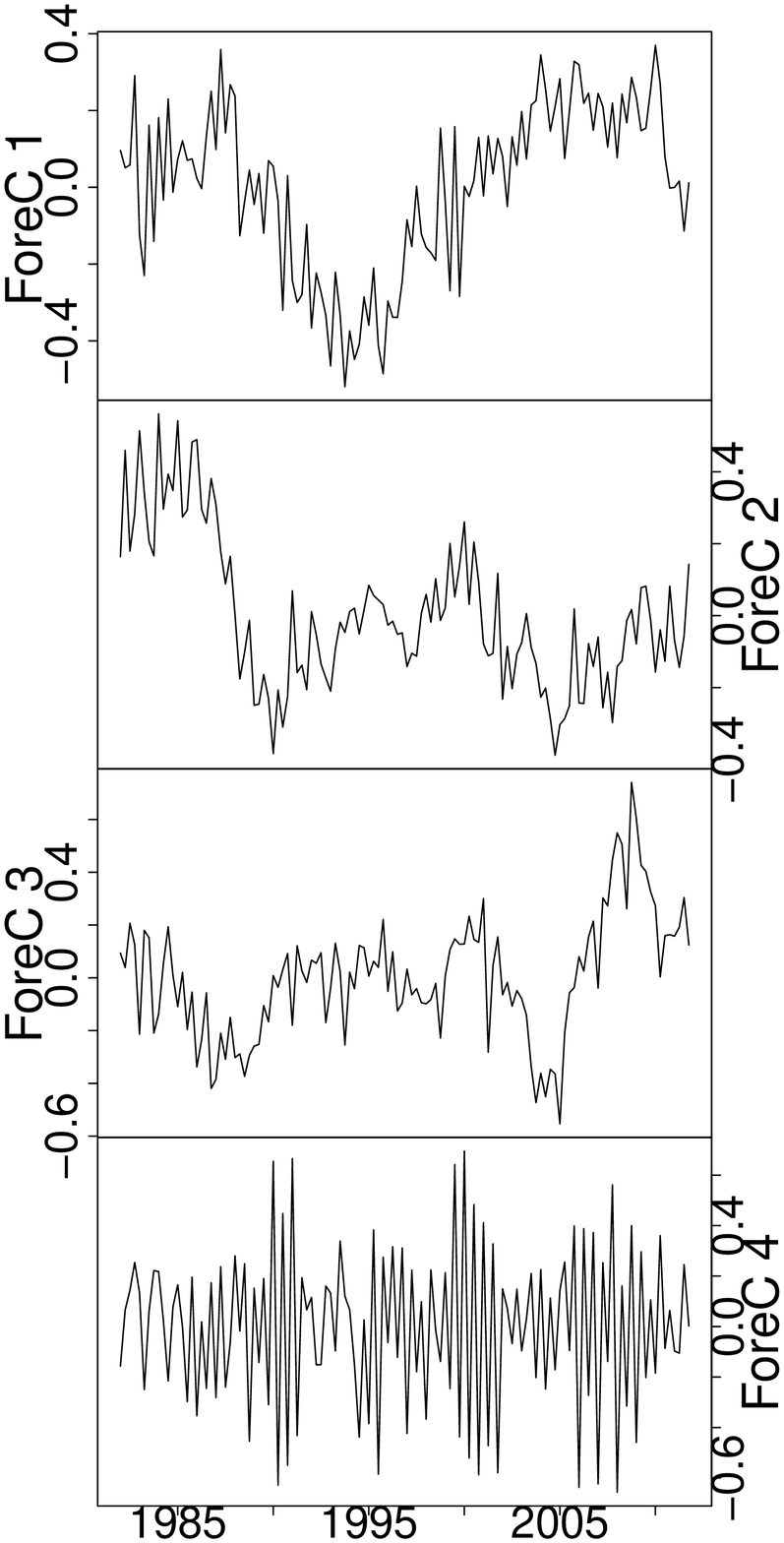}
        \caption{\label{fig:US_income_ForeCs} ForeCs}
\end{subfigure}%
\begin{subfigure}[t]{\figWidth\textwidth}
        \centering
        \includegraphics[width=\textwidth]{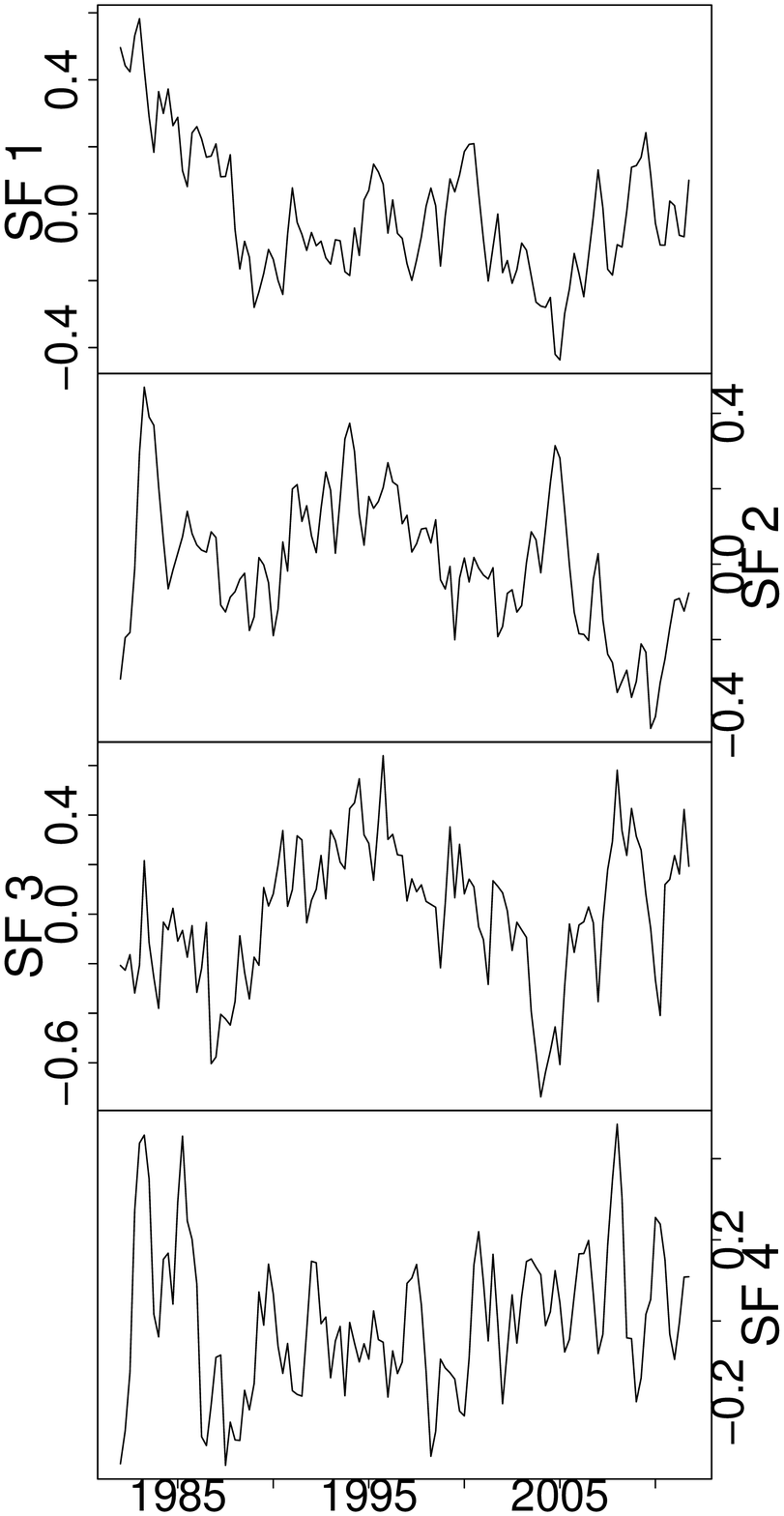}
        \caption{\label{fig:US_income_SFs} SFs}
\end{subfigure}
\begin{subfigure}[t]{\figWidth\textwidth}
        \centering
        \includegraphics[width=\textwidth]{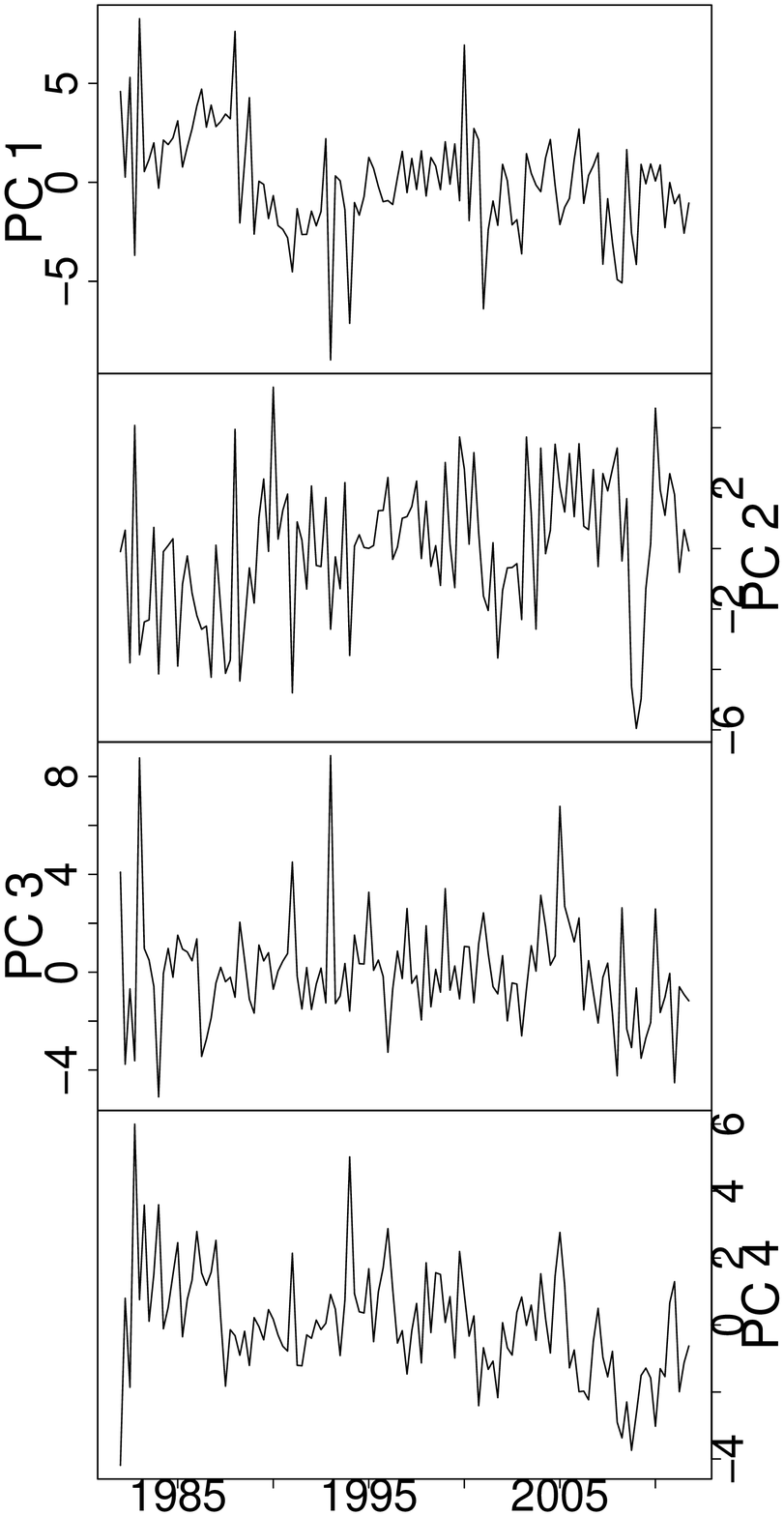}
        \caption{\label{fig:US_income_PCs} PCs}
\end{subfigure}

\renewcommand{\figWidth}{0.48}
\begin{subfigure}[t]{\figWidth\textwidth}
        \centering
        \includegraphics[width=\textwidth]{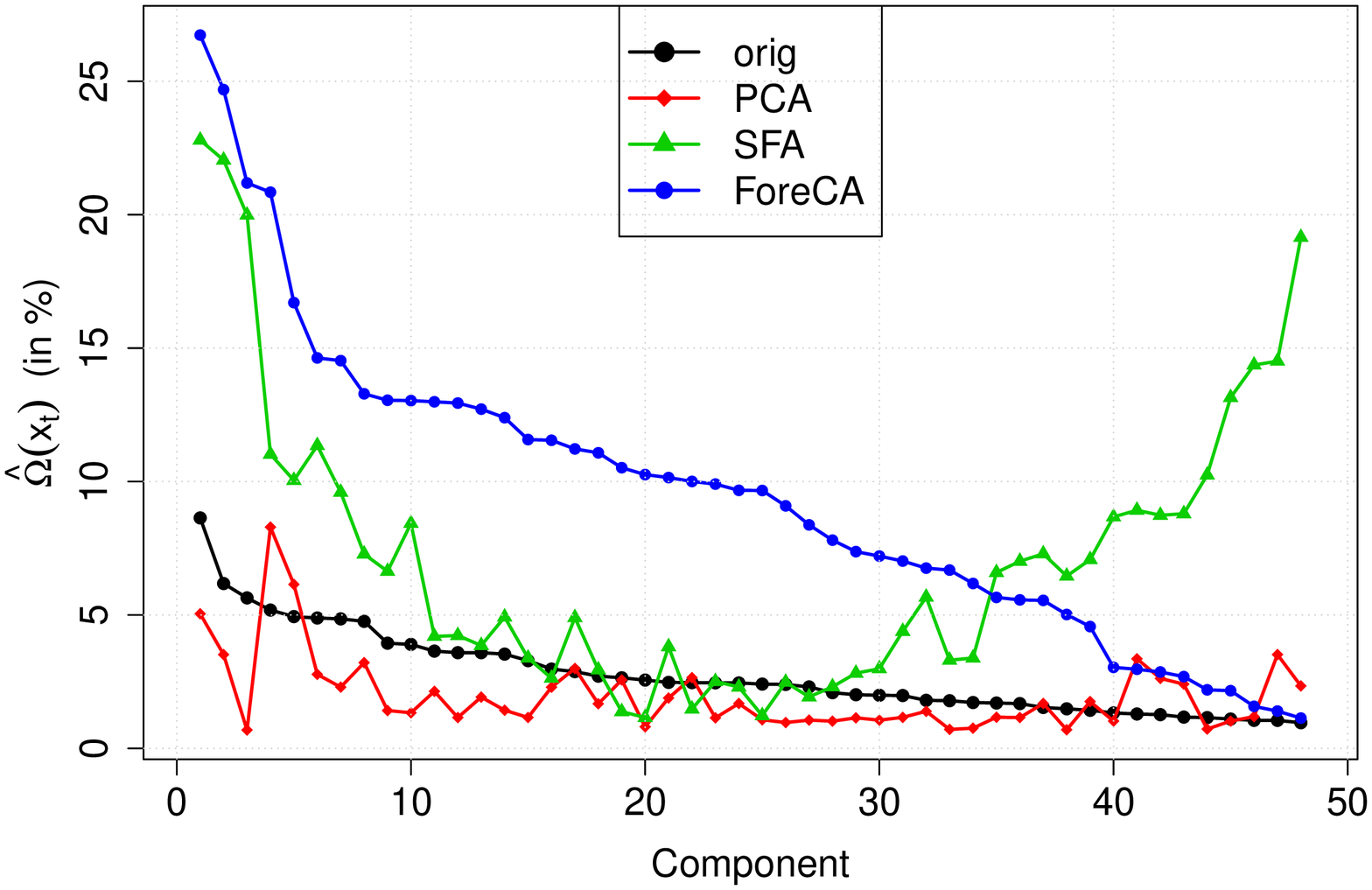}
        \caption{\label{fig:US_income_Omega_Orig_PCA_SFA_ForeCA} scree-plot of $\widehat{\Omega}(\cdot)$.}
\end{subfigure}%
\caption[US income ForeCA results.]{\label{fig:US_income_ForeCA_results} PCA, SFA, and ForeCA on US income data.}
\end{figure}

\section{Related Work}
\label{sec:related_work}
Using predictability to separate signals is not new. 

In the classic time series literature \citet{BoxTiao77_CCAtimeseries} introduced canonical analysis and measure predictive power by the residual variance of fitting vector auto-regression (VAR) models.  Recently \citet{MattesonTsay11_DOC_timeseries} propose another DR technique that blends PCA and ICA by separating signals to the extent of fourth moments (but not higher). 

\citet{Stone01_BSS_TemporalPredictability} use predictability as a contrast function for blind source separation (BSS).  While their approach is similar to ours, it relies on subjective measures of ``short'' and ``long'' term moving averages, which are then used to produce actual forecasts.  

Much work in BSS \citep{Gomezetal10_BSS_minentropy, LinAdali10_BSS_minentropy}, especially ICA, focuses on minimizing entropy rate. The \emph{entropy rate} $\mathcal{H}(y_t) = \lim_{t  \rightarrow \infty} H(y_t \mid y_{t-1}, y_{t-2}, \ldots)$ of a Gaussian process is related to the spectrum via \citep[][p.\ 417]{CoverThomas91_InformationTheory}
\begin{align}
\mathcal{H}(y_t) 
\label{eq:entropy_rate}
&= \frac{1}{2} \log 2 \pi e + \frac{1}{4 \pi} \int_{-\pi}^{\pi} \log S_y(\lambda) d \lambda.
\end{align} 
However, these approaches require VAR model fits and/or numerical optimization.

On the contrary, the ForeCA measure $\Omega(y_t)$ is based on information-theoretic uncertainty and is an inherent property of the stochastic process $y_t$. We believe that this makes $\Omega(y_t)$ a more principled measure of forecastability than model-dependent measures. Furthermore, it can be estimated quickly using data-driven, nonparametric techniques.

It is important to point out that \emph{spectral entropy}, i.e., differential entropy of \eqref{eq:spectral_entropy}, is neither equal nor proportional to the entropy rate in \eqref{eq:entropy_rate}.  For particular processes they coincide (e.g., for an $AR(1)$; \citet{Gibson94_spectralentropy_interpretation}), but in general they don't. They measure \emph{different} properties of the signal.  Thus ICA algorithms based on entropy rate minimization do not yield the same results as ForeCA.  In fact, the ForeCA measure can be used to rank ICs by decreasing forecastability.

\citet{Cardoso04_GaussianityCorrelationDependence} gives an excellent account of the intertwined relations between Gaussianity, autocorrelation, and dependence in multivariate time series and their effect on objective functions for BSS.  Exactly because of this tangle, we only consider frequency properties of the signal and not entropy rate -- since for forecasting the distribution itself is of minor importance compared to the temporal dependence.


\section{Discussion}
\label{sec:discussion}

I introduce Forecastable Component Analysis (ForeCA), a new dimension reduction technique for multivariate time series. Contrary to other popular methods -- such as PCA or ICA -- ForeCA takes temporal dependence into account and actively searches for the most forecastable subspace.  ForeCA minimizes the entropy of the spectral density: lower entropy implies a more forecastable signal. The optimization problem has an iterative, yet fast analytic solution, and provably leads to a (local) optimum. 

While SFA is a good approximation (maximizing lag $1$ correlation), real world signals often have more complex correlation structure. The here proposed ForeCA can automatically detect arbitrary autocorrelation structure using nonparametric estimators. Applications to financial and macro-economic data demonstrate that ForeCA is better than PCA and SFA at finding the most predictable signals, and can also be used for time series classifications.

\clearpage
\addcontentsline{toc}{subsection}{References}
\bibliographystyle{chicago}	
\bibliography{PhD_thesis}

\cleardoublepage

\appendix

\begin{center}
\textbf{ APPENDIX -- SUPPLEMENTARY MATERIAL }
\end{center}

\section{Proofs}

\begin{proof}[Proof of Properties \ref{prop:Omega}\ref{item:Omega_comb_less_max} (Max-subadditivity)]

Consider the linear combination of two stationary processes $x_t$ and $y_t$,
\begin{equation*}
z_t = \alpha x_t + \beta y_t
\end{equation*}

It holds $\V z_t = \alpha^2 \sigma_x^2 + \beta^2 \sigma_y^2 + 2 \alpha \beta cov(x_t, y_t)$. If $x_t$ and $y_s$ are uncorrelated for all $s \neq t$ then
$\V z_t = \alpha^2 + \beta^2$ (if $x_t$ and $y_t$ are both unit-variance processes). To have a unit-variance sum we need $\beta = \sqrt{1- \alpha^2}$:

The spectrum of $z_t$ equals (if they are uncorrelated)
\begin{equation}
\label{eq:sum_spectra_less_square}
S_z(\lambda) = \alpha^2 S_x(\lambda) + \beta^2 S_y(\lambda) \leq \alpha S_x(\lambda) + \beta S_y(\lambda)
\end{equation}

It therefore holds
\begin{align*}
\Omega(z_t) 
	&= \Omega(\alpha x_t + \beta y_t) \\
	& = 1 - H\left( \alpha^2 S_x(\lambda) + \beta^2 S_y(\lambda)  \right) \\
	& \leq 1 - \left( \alpha^2 H(S_x(\lambda)) + \beta^2  H(S_y(\lambda)) \right),
\end{align*}
where the last inequality follows by \eqref{eq:sum_spectra_less_square}.  Plugging in $\beta = \sqrt{1- \alpha^2}$ 
gives
\begin{align*}
& \alpha^2 - \alpha^2 H(S_x(\lambda)) + (1-\alpha^2) - (1-\alpha^2) H(S_y(\lambda)) \\
	& = \alpha^2 (1 - H(S_x(\lambda))) + (1-\alpha^2) (1 - H(S_y(\lambda))) \\
	& = \alpha^2 \Omega(x_t) + (1-\alpha^2) \Omega(y_t) \\
	& \leq \max \left( \Omega(x_t),  \Omega(y_t) \right),
\end{align*}
which completes the proof.

\end{proof}

\begin{proof}[Proof of Proposition \ref{prop:semi_def}]
For every $\mathbf{w}$,
\begin{align*}
\transpose{\mathbf{w}} \overline{\widehat{S}^{(i)}_U} \mathbf{w} & = - \frac{1}{T} \sum_{j=0}^{T-1} \transpose{\mathbf{w}} \left[ \widehat{S}_U(\omega_j) \cdot \ell(\mathbf{w}_i; \omega_j) \right] \mathbf{w}\\
&= \frac{1}{T} \sum_{j=0}^{T-1} \underbrace{ - \ell(\mathbf{w}_i; \omega_j)}_{\geq 0} \cdot \underbrace{\transpose{\mathbf{w}} \widehat{S}_U(\omega_j) \mathbf{w}}_{\geq 0} \geq 0.
\end{align*}
\end{proof}

\end{document}